\documentclass[10pt,prd,aps,eqsecnum,superscriptaddress,nofootinbib,notitlepage,longbibliography,final]{revtex4-2}
\usepackage[T1]{fontenc}
\usepackage[utf8]{inputenc}
\usepackage{lmodern}
\usepackage{amsmath,amssymb,amsthm,mathtools}
\usepackage{microtype}
\usepackage{comment}
\usepackage[final,pagebackref]{hyperref}
\hypersetup{hidelinks,pdftitle={Quantum Gibbs states are quasi-locally Markovian},pdfauthor={Tai Hsuan Yang}}
\renewcommand*{\backrefalt}[4]{%
 \ifcase #1\relax
 \or\space(Cited on page~#2)%
 \else\space(Cited on pages~#2)%
 \fi}
\backrefparscanfalse
\let\OriginalBibitemShut\BibitemShut
\renewcommand{\BibitemShut}[1]{\begingroup\let\par\relax\backrefprint\endgroup\OriginalBibitemShut{#1}}
\allowdisplaybreaks[1]
\newtheorem{theorem}{Theorem}[section]

\newtheorem{lemma}[theorem]{Lemma}
\newtheorem{corollary}[theorem]{Corollary}
\theoremstyle{definition}

\theoremstyle{remark}
\newtheorem{remark}[theorem]{Remark}
\DeclareMathOperator{\Tr}{Tr}
\DeclareMathOperator{\supp}{supp}
\DeclareMathOperator{\diam}{diam}

\DeclareMathOperator{\ad}{ad}
\newcommand{\one}{\mathbf 1}
\newcommand{\R}{\mathbb R}
\newcommand{\Z}{\mathbb Z}

\begin{document}
\title{Improved estimate of local Markovianity for quantum Gibbs states}
\author{Tai Hsuan Yang}
\email{ty39@illinois.edu}
\affiliation{Department of Physics and Institute for Condensed Matter Theory, University of Illinois Urbana-Champaign}
\date{September 28, 2026}
\begin{abstract}
We establish improved locally Markovian properties of quantum Gibbs states at every fixed positive temperature.
For finite-range interactions of bounded interaction degree, the conditional mutual information (CMI) decays exponentially with separation, with a prefactor exponential in the boundary size of only one of the two regions and independent of the size of the other region and the total system.
This improves the exponential volume dependence of the local Markov bounds established by Chen and Rouz\'e \cite{Chen2025}.
For exponentially and stretched-exponentially decaying interactions, we obtain stretched-exponential CMI decay with a boundary-dependent prefactor.
For power-law interactions, we prove inverse-logarithmic or algebraic decay with polynomial region-size dependence.
Our proof is static, in contrast to the dynamical approach of Chen and Rouz\'e.
We compare the Gibbs state with a product state obtained by removing boundary interactions, reducing the CMI bound to a relative-entropy loss that we control using locality estimates.
\end{abstract}
\maketitle

\noindent\textbf{AI usage statement.} The results presented in this paper were found by GPT-6 Astra Pro.  The author verifies and organizes the proof and takes full responsibility for the correctness of the paper. 
Chat history can be found \href{https://chatgpt.com/share/6abaa570-f428-83e9-95cb-9c0ec2a2661d}{here}.

\tableofcontents

\section{Introduction}\label{sec:gm:introduction}

Local interactions constrain thermal correlations most directly through conditional independence. 
In a classical Gibbs distribution, conditioning on a region $B$ that separates $A$ from $C$ makes the variables in $A$ and $C$ independent. 
The Hammersley--Clifford theorem expresses this connection between positive Markov random fields and Gibbs distributions with clique-local interactions \cite{Hammersley1971MarkovFO}. 
It is distinct from decay of ordinary correlations: a classical Gibbs distribution retains its Markov property even at a thermal phase transition.  
The quantum analogue asks how much conditional dependence can survive across a wide shielding region when the interaction terms need not commute.

Let $\Lambda=A\sqcup B\sqcup C$ be a finite set of sites partitioned into three regions, and write $|X|$ for the number of sites in a region $X$.  For a quantum state $\rho_{ABC}$ on this partition, conditional dependence is measured by the conditional mutual information (CMI),
\begin{equation}\label{eq:gm:intro-cmi}
 I(A:C\mid B)_\rho
 =S(\rho_{AB})+S(\rho_{BC})-S(\rho_B)-S(\rho_{ABC}),
 \qquad S(\omega)=-\Tr(\omega\log\omega),
\end{equation}

where $\rho_X=\Tr_{\Lambda\setminus X}\rho$ denotes the reduced state on $X$ and logarithms are natural.  Vanishing CMI is equivalent to exact recovery of $A$ from $B$: there is a completely positive trace-preserving map $\mathcal R_{B\to AB}$ with $\rho_{ABC}=(\mathcal R_{B\to AB}\otimes\operatorname{id}_C)(\rho_{BC})$ \cite{HaydenJozsaPetzWinter2004}.  The Fawzi--Renner theorem makes this interpretation quantitative, bounding the recovery error in terms of the CMI \cite{FawziRenner_2015}. 
Gibbs states of commuting local Hamiltonians have the exact Markov property whenever $B$ shields $A$ from $C$ in the interaction graph \cite{BrownPoulin2012}.  Noncommuting Hamiltonians generally require an approximate statement: for $\rho_{\Lambda,\beta}=e^{-\beta H_\Lambda}/\Tr e^{-\beta H_\Lambda}$, does the CMI decay as the separation $r=d(A,C)$ increases?  Here $H_\Lambda$ is the Hamiltonian on $\Lambda$, $\beta>0$ is the inverse temperature, and $d(A,C)$ is the minimum distance between sites in $A$ and $C$. 
We write $n=|\Lambda|$ for the total number of sites and $D$ for the spatial dimension.  

\begin{table}[tp!]
\centering
\setlength{\tabcolsep}{4pt}
\renewcommand{\arraystretch}{1.05}
\footnotesize
\begin{tabular}{|p{2.2cm}|p{4.0cm}|p{5.2cm}|p{4.1cm}|}
\hline
\raggedright \textbf{Reference} & \raggedright \textbf{State and regime} & \raggedright \textbf{Upper bound on $I(A:C\mid B)$} & \raggedright \textbf{Comments} \tabularnewline
\hline\hline
\raggedright Brown--Poulin \cite{BrownPoulin2012} & \raggedright Commuting Gibbs states; all temperatures & \raggedright $0$ & \raggedright  \tabularnewline
\hline
\raggedright Kato--Brand\~ao \cite{KatoBrandao2019} & \raggedright 1D finite-range Gibbs states; all temperatures & \raggedright $C_\beta e^{-c_\beta\sqrt r}$ & \raggedright Restricted to $D=1$. \tabularnewline
\hline
\raggedright Kuwahara \cite{Kuwahara2025} & \raggedright Finite-range Gibbs states; all temperatures; $D>1$ & \raggedright $C_\beta e^{C_\beta(|A|+|C|)}e^{-c_\beta r/\log(2+r)}$ & \raggedright Exponential dependence on both $|A|$ and $|C|$. \tabularnewline
\cline{2-4}
\raggedright & 1D finite-range Gibbs states; all temperatures & \raggedright $C_\beta e^{-c_\beta r}$ & \raggedright Restricted to $D=1$. \tabularnewline
\hline
\raggedright Kato--Kuwahara \cite[Thms.~1+2]{KatoKuwahara2025} & \raggedright Finite-range Gibbs states; rapid mixing of the quantum Gibbs samplers & \raggedright $\operatorname{poly}(n)e^{-c_\beta r^{1/(D+3)}}$ & \raggedright Proof through Belief-propagation channel \tabularnewline
\cline{1-4}
\raggedright Kato--Kuwahara \cite[Thms.~1+3]{KatoKuwahara2025} & \raggedright Finite-range Gibbs states; uniform clustering for all subset Hamiltonians & \raggedright $\operatorname{poly}(n)e^{-c_\beta r^{1/D}}$ & \raggedright Condition under uniform-clustering. \tabularnewline
\hline
\raggedright Scalet et al.\ \cite[Lemma~1]{Scalet2025} & \raggedright Gibbs states with exponentially decaying effective interactions & \raggedright $C_\beta n e^{-c_\beta r}$ & \raggedright Conditional on effective-interaction locality. \tabularnewline
\hline
\raggedright Bakshi et al.\ \cite[Thm.~2.3]{BakshiLiuMoitraTang2025} & \raggedright High temperature finite-range Gibbs states & \raggedright $C_\beta|A||C|e^{-c_\beta r}$ & \raggedright Quantum Dobrushin condition. \tabularnewline
\hline
\raggedright Chen--Rouz\'e \cite{Chen2025} & \raggedright Finite-range Gibbs states; all temperatures & \raggedright $C_\beta|A||C| e^{C_\beta\min\{|A|,|C|\}-c_\beta r}$ & \raggedright Time-averaged Gibbs sampling; First all-temperature result in all dimensions. \tabularnewline
\hline

\raggedright Rosa-Ruiz et al.\ \cite[Thm.~3.3]{rosaruiz2026staticfeaturesmixingshort} & \raggedright Fixed points of primitive, k-local, frustration-free rapid-mixing Lindbladians with $F$ decay & \raggedright $p(|A|,|C|)\,r^{\nu D/2}
 F(r)^{\lambda/[2(\lambda+v)]}$
\par $\displaystyle\sup_{\rho}\bigl\|e^{t\mathcal L_\Lambda^*}(\rho)-\sigma_\Lambda\bigr\|_1 \displaystyle\le C|\Lambda|^\nu e^{-\lambda t}$ & \raggedright General fixed points of the Lindbladian. \tabularnewline
\cline{1-4}
\raggedright Rosa-Ruiz et al.\ \cite[Cor.~3.4]{rosaruiz2026staticfeaturesmixingshort} & \raggedright Same Lindbladian assumptions; $F_\alpha(r)=(1+r)^{-\alpha}$ & \raggedright $p(|A|,|C|)(1+r)^{-\eta}$,
\newline $\eta=\dfrac{\lambda\alpha}{2(\lambda+v)}-\dfrac{\nu D}{2}>0$ & \raggedright Requires $\alpha>\nu D(\lambda+v)/\lambda$. \tabularnewline
\cline{1-4}
\raggedright Rosa-Ruiz et al.\ \cite[Thm.~4.6]{rosaruiz2026staticfeaturesmixingshort} & \raggedright $k$-local power-law Gibbs states; $\alpha>D$; all temperatures & \raggedright $\widetilde O\!\left(|C|e^{C_\beta|A|}r^{-\eta_\beta}\right)$,
\newline $\eta_\beta=\lambda'_\beta(\alpha-D)/2$ & \raggedright Extends local Markovianity to power-law interactions. \tabularnewline
\hline\hline
\raggedright This work, Theorem~\ref{thm:gm:local} & \raggedright Finite-range Gibbs states; all temperatures & \raggedright $C_\beta e^{C_\beta|\partial_{\mathrm e}A|-c_\beta r}$ & \raggedright Exponential boundary dependence. \tabularnewline
\hline
\raggedright This work, Theorem~\ref{thm:gm:quasilocal} & \raggedright Aggregate decay $e^{-\mu r^\theta}$; $0<\theta\le1$; $D\ge2$; unrestricted body size; all temperatures & \raggedright $\begin{gathered}C_\beta\exp\!\bigl[C_\beta(1+|\partial_{\mathrm e}A|)^{\theta/D}\\{}-c_\beta(1+r)^{\theta/(D+1-\theta)}\bigr]\end{gathered}$ & \raggedright Unrestricted body size; exponential interactions included at $\theta=1$. \tabularnewline
\hline

\raggedright This work, Theorem~\ref{thm:gm:powerlaw} & \raggedright Power-law Gibbs states; unrestricted body size; $D<\zeta\le2D$; all temperatures & \raggedright $C_\beta(1+|A|)^{b_\zeta+1}[\log(e+r)]^{-b_\zeta}$ & \raggedright At $\zeta=2D$, multiply by $[\log(e+\log(e+r))]^2$. No $|C|$ or total-volume dependence. \tabularnewline
\hline
\raggedright This work, Theorem~\ref{thm:gm:powerlaw} & \raggedright Power-law Gibbs states; unrestricted body size; $\zeta>2D$; all temperatures & \raggedright $C_\beta(1+|A|)^{b_\zeta+1}(1+r)^{-\nu_\zeta}$ & \raggedright $\nu_\zeta$ in \eqref{eq:gm:powerlaw-exponents}; no logarithmic correction. \tabularnewline
\hline
\end{tabular}
\caption{Comparison of CMI bounds for full partitions $\Lambda=A\sqcup B\sqcup C$.  Here $r=d(A,C)$, $n=|\Lambda|$, $D$ is the spatial dimension, and $|\partial_{\mathrm e}A|$ counts boundary edges. 
Bounds are stated at fixed $\beta>0$ and fixed interaction, mixing, and clustering parameters; constants are independent of region sizes, with all size dependence displayed explicitly.  
Polynomial distance factors are absorbed into exponential rates, and $\widetilde O$ suppresses logarithmic distance factors.  
The function $p$ is polynomial, $F$ is a spatial decay profile, and $k$-local means that each Hamiltonian or generator term acts on at most $k$ sites. 
In the rapid-mixing estimate, $\mathcal L_\Lambda^*$ generates evolution of states, $\sigma_\Lambda$ is its stationary state, and $\|\cdot\|_1$ is the trace norm; $\nu$ is the polynomial volume exponent, $\lambda$ the temporal decay rate, and $v$ the Lieb--Robinson velocity.  
The exponent $\alpha$ specifies power-law decay in the cited results, and $\lambda'_\beta>0$ is the temperature-dependent rate entering their Gibbs-state bound. 
In our stretched-exponential row, $\mu>0$ is the decay rate and $\theta$ the stretching exponent.  
Pairwise and local bounds have exponential dependence on both $|A|$ and $|C|$ and on only one of these sizes, respectively; global bounds have at most polynomial size dependence.  
the rate exponents are given in Theorem~\ref{thm:gm:powerlaw}.}
\label{tab:gm:markov-properties}
\label{tab:gm:markov-properties-longrange}
\end{table}

The first general results for noncommuting systems were obtained in one dimension.  Araki established exponential decay of correlations in Gibbs states of translation-invariant, finite-range quantum spin chains at every positive temperature \cite{Araki1969}. 
Using one-dimensional thermal clustering and quantum belief propagation \cite{Hastings2007QBP,Kim2012PerturbativeTEE}, Kato and Brand\~ao proved the stretched-exponential CMI bound $I(A:C\mid B)_\rho\le C_\beta e^{-c_\beta\sqrt r}$ for finite-range quantum spin chains at every fixed positive temperature, without assuming translation invariance and with constants independent of the region sizes and total chain length \cite{KatoBrandao2019}.  Kuwahara subsequently obtained exponential decay in the buffer length in one dimension, and an arbitrary-temperature pairwise Markov bound in higher dimensions \cite{Kuwahara2025}.  The latter uses locality estimates for effective Hamiltonians of reduced states, with an exponential dependence on both $|A|$ and $|C|$.

In a major advance, Chen and Rouz\'e established that Gibbs states of finite-range Hamiltonians with bounded interaction degree are locally Markovian at arbitrary positive temperature and in any spatial dimension \cite{Chen2025}.
Their proof constructs a recovery map by time-averaging the evolution generated by the quasi-local Kubo--Martin--Schwinger (KMS) detailed-balanced Gibbs sampler \cite{ChenKastoryanoBrandaoGilyen2023,ChenKastoryanoGilyen2023,GilyenChenDoriguelloKastoryano2024,ChenEtAl2025ThermalSimulation,DingLiLin2025}.
At the same time, Kato and Kuwahara used quasi-local quantum belief-propagation channels to prove CMI bounds of the form $\operatorname{poly}(n)e^{-c_\beta r^{1/(D+3)}}$ under rapid mixing of the KMS detailed-balanced Gibbs sampler and $\operatorname{poly}(n)e^{-c_\beta r^{1/D}}$ under uniform clustering \cite{KatoKuwahara2025}.
With the assumption of high temperature, Bakshi, Liu, Moitra, and Tang later obtain the exponential bound $C_\beta|A||C|e^{-c_\beta r}$ for local Hamiltonians with bounded interaction and growth parameters, 
using a quantum Dobrushin condition and rapid mixing \cite[Theorem~2.3]{BakshiLiuMoitraTang2025}\cite{RouzeFrancaAlhambra2026Optimal}. 
More recently, Rosa-Ruiz et al.\ \cite{rosaruiz2026staticfeaturesmixingshort} proved CMI decay with polynomial region-size dependence for fixed points of primitive, frustration-free, rapidly mixing $k$-local Lindbladians: exponential for short-range interactions and algebraic for sufficiently fast power-law decay.
They also extended local Markovianity to $k$-local power-law Gibbs states with $\alpha>D$ at arbitrary positive temperature, obtaining algebraic CMI decay up to logarithmic factors with a prefactor exponential in one region size and linear in the other.

By bounding the relative-entropy loss directly without constructing a recovery map, we replace the exponential \emph{volume} dependence by exponential dependence on the \emph{boundary} of $A$ or, by symmetry, $C$,with no size dependence on the other region or the total system.  Write $H_\Lambda=\sum_{X\subseteq\Lambda}\Phi(X)$, where $\Phi(X)$ is supported on $X$, and let $A^c=\Lambda\setminus A$.  With $\|\cdot\|_\infty$ denoting the operator norm, for finite-range interactions we prove
\begin{equation}\label{eq:gm:intro-bound}
 I(A:C\mid B)_{\rho_{\Lambda,\beta}}
 \le C_\beta e^{C_\beta g_{A\mid A^c}-c_\beta r},
 \qquad
 g_{A\mid A^c}=\sum_{\substack{X\subseteq\Lambda\\X\cap A\ne\varnothing,\ X\cap A^c\ne\varnothing}}
 \|\Phi(X)\|_\infty,
\end{equation}

Here $g_{A\mid A^c}$ is the total interaction strength across the cut.  On the lattice $\Z^D$, the cut strength $g_{A\mid A^c}$ is bounded by a constant times $|\partial_{\mathrm e}A|$, the number of nearest-neighbor edges crossing the boundary.  Thus the exponential factor improves from $e^{O_\beta(|A|)}$ to $e^{O_\beta(|\partial_{\mathrm e}A|)}$. 
We also extend the argument to quasi-local interactions. 
Under the aggregate decay assumption $\sum_{X\ni x,y}\|\Phi(X)\|_\infty\le J_{\mu,\theta}e^{-\mu d(x,y)^\theta}$ of Theorem~\ref{thm:gm:quasilocal}, 
where $\mu>0$ and $J_{\mu,\theta}<\infty$ are fixed and $0<\theta\le1$, we obtain stretched-exponential CMI decay in $D\ge2$ 
without a $k$-local assumption on the interaction terms.  
For power-law interactions satisfying
$\sum_{X\ni x,y}\|\Phi(X)\|_\infty\le J_\zeta(1+d(x,y))^{-\zeta}$ with fixed $J_\zeta<\infty$ and decay exponent $\zeta>D$,
Theorem~\ref{thm:gm:powerlaw} gives inverse-logarithmic CMI decay for $D<\zeta\le2D$ (with a squared iterated-logarithm correction at $\zeta=2D$) and algebraic decay for $\zeta>2D$.
This differs in scope from the $k$-local power-law Gibbs result in Ref.~\cite{rosaruiz2026staticfeaturesmixingshort}: our support sizes are unrestricted, while the algebraic theorem uses the stronger sufficient threshold $\zeta>2D$.
Table~\ref{tab:gm:markov-properties} compares these results with previous CMI bounds, listing their assumptions, distance decay, and region-size dependence.  It highlights the boundary-dependent prefactors in our finite-range and stretched-exponential bounds and the polynomial region-size dependence in our power-law bounds.

Our approach is static: we compare the Gibbs state with the reference state obtained by removing interactions across the cut $A\mid BC$. 
The reference state is a product across this cut, and the CMI is bounded by the decrease in relative entropy between the two states when $C$ is traced out. 
A variational bound based on a resolvent formula reduces the estimate of this loss to spatial localization, controlled by Lieb--Robinson bounds \cite{NachtergaeleSimsYoung2019,ChenLucasYin2023}, and spectral estimates for the relative modular operator.  
For finite-range interactions, short-imaginary-time estimates control the prefactor in terms of the interaction strength across the cut. 
For the stretched-exponential and algebraic power-law bounds, we first truncate long interactions within a shielding neighborhood around $A$ and control the resulting change in CMI by Gibbs interpolation.  
We then combine spatial localization near $A$ with spectral bounds derived from the decay of interaction strengths with support size to derive the final CMI estimates. 

Section~\ref{sec:gm:setting} states the assumptions and main results, and Section~\ref{sec:gm:modular} develops the comparison with the cut Gibbs state and the variational bound.  Sections~\ref{sec:gm:local-proof}, \ref{sec:gm:quasilocal-proof}, and~\ref{sec:gm:powerlaw-proof} prove the finite-range, stretched-exponential, and power-law bounds, respectively.  Section~\ref{sec:gm:discussion} discusses their scope and open questions.  Appendices~\ref{app:gm:gronwall}, \ref{app:gm:gap-integral}, and~\ref{app:gm:conditional-logarithm} provide auxiliary proofs, and Appendix~\ref{app:gm:notation} collects the notation.

\section{Setting and main results}\label{sec:gm:setting}

Let $\Gamma$ be a set of sites with distance $d(x,y)$, and assign a finite-dimensional Hilbert space $\mathcal H_x$ to each site $x$.  On $\Z^D$, we use the nearest-neighbor graph metric $d(x,y)=|x-y|_1$, where $|z|_1:=\sum_{j=1}^D|z_j|$ for $z=(z_1,\ldots,z_D)\in\Z^D$.  For a finite volume $\Lambda\Subset\Gamma$, all operators are understood to act on $\mathcal H_\Lambda=\bigotimes_{x\in\Lambda}\mathcal H_x$, with identities on omitted tensor factors.  We write $XY=X\cup Y$ for disjoint regions, $X^c=\Lambda\setminus X$, and $d(X,Y)=\min_{x\in X,y\in Y}d(x,y)$ for nonempty finite sets.  We use the ordinary, unnormalized trace, natural logarithms, the operator norm $\|O\|_\infty$, the trace norm $\|O\|_1=\Tr\sqrt{O^\dagger O}$, and the Hilbert--Schmidt inner product $\langle X,Y\rangle=\Tr(X^\dagger Y)$, with norm $\|O\|_2^2=\langle O,O\rangle$.  Entropy and CMI are defined in \eqref{eq:gm:intro-cmi}; the relative entropy is
\begin{equation}\label{eq:gm:entropy-definitions}
 D(\rho\|\sigma)=\Tr\rho(\log\rho-\log\sigma).
\end{equation}

We consider the Gibbs state
\begin{equation}\label{eq:gm:gibbs}
 H_\Lambda=\sum_{X\subseteq\Lambda}\Phi(X),
 \qquad
 \rho_{\Lambda,\beta}=\frac{e^{-\beta H_\Lambda}}{Z_\Lambda},
 \qquad Z_\Lambda=\Tr e^{-\beta H_\Lambda},\qquad \beta>0,
\end{equation}

where $\Phi(X)=\Phi(X)^\dagger$ is supported on $X$.  An empty-support scalar term can be omitted.  All Gibbs states and their marginals are faithful in finite volume.

Throughout the main theorems,
\begin{equation}\label{eq:gm:full-partition}
 \Lambda=A\sqcup B\sqcup C,
 \qquad r=d(A,C),
\end{equation}

is a \emph{full} partition, with $A$ and $C$ nonempty.  Thus $B$ contains every site of $\Lambda$ outside $A\cup C$.  This hypothesis is essential to the cut-Hamiltonian argument; a marginal with an arbitrary unobserved exterior is not silently identified with a local Gibbs state.

For each interaction class, define the cut interaction and its total strength by
\begin{equation}\label{eq:gm:cut-strength}
 V=\sum_{\substack{X\subseteq\Lambda\\ X\cap A\ne\varnothing,\ X\cap A^c\ne\varnothing}}\Phi(X),
 \qquad
 g_{A\mid A^c}=\sum_{\substack{X\subseteq\Lambda\\ X\cap A\ne\varnothing,\ X\cap A^c\ne\varnothing}}\|\Phi(X)\|_\infty.
\end{equation}

In particular, $\|V\|_\infty\le g_{A\mid A^c}$.  When interactions are indexed by terms $h_\gamma$ instead of supports $X$, the same definitions are taken term by term.
The bounds below quantify approximate conditional independence across the full partition and, through the recovery theorem, imply approximate recovery of $A$ by a channel acting on $B$ and leaving $C$ unchanged.

\begin{theorem}[Finite-range interactions]\label{thm:gm:local}
Suppose
\begin{equation}\label{eq:gm:finite-range-assumptions}
 H_\Lambda=\sum_\gamma h_\gamma,
 \qquad \|h_\gamma\|_\infty\le J,
 \qquad \diam(\supp h_\gamma)\le R_0.
\end{equation}
Assume that each term overlaps at most $\mathfrak d$ terms, counting itself.  The bounds $J>0$, $R_0$, and $\mathfrak d$ are uniform in $\Lambda$.  For every fixed $\beta>0$, there exist $C_\beta,c_\beta>0$, depending only on these interaction parameters and $\beta$, such that
\begin{equation}\label{eq:gm:local-main}
 I(A:C\mid B)_{\rho_{\Lambda,\beta}}
 \le
 C_\beta\exp\!\left(C_\beta g_{A\mid A^c}-c_\beta r\right).
\end{equation}
On $\Z^D$ with its nearest-neighbor graph metric, let $\partial_{\mathrm e}A$ denote the set of unoriented nearest-neighbor edges with exactly one endpoint in $A$.  Then $g_{A\mid A^c}\le C |\partial_{\mathrm e}A|$, with $C$ depending only on the local interaction parameters, and therefore
\begin{equation}\label{eq:gm:local-boundary}
 I(A:C\mid B)_{\rho_{\Lambda,\beta}}
 \le C_\beta e^{C_\beta |\partial_{\mathrm e}A|-c_\beta r}.
\end{equation}

\end{theorem}

\begin{theorem}[Stretched-exponential interactions of unrestricted body size]\label{thm:gm:quasilocal}
Let $\Gamma=\Z^D$, $D\ge2$, with the $\ell^1$ metric and uniformly bounded on-site dimensions.  Suppose $\mu>0$, $0<\theta\le1$, and
\begin{equation}\label{eq:gm:F-assumption}
 J_{\mu,\theta}:=\sup_{x,y\in\Z^D}
 e^{\mu d(x,y)^\theta}
 \sum_{\substack{X\Subset\Z^D\\x,y\in X}}\|\Phi(X)\|_\infty
 <\infty.
\end{equation}
The supremum includes $x=y$.  Thus this is the aggregate pair bound
$\sum_{X\ni x,y}\|\Phi(X)\|_\infty\le J_{\mu,\theta}e^{-\mu d(x,y)^\theta}$.
No uniform bound on $|X|$ or on the number of terms meeting a site is assumed.
For every fixed $\beta>0$, there are constants $C_\beta,c_\beta>0$ depending only on
\[
 \beta,\ D,\ \mu,\ \theta,\ J_{\mu,\theta},
\]
and, for $C_\beta$, the fixed bound on the on-site dimensions, such that
\begin{equation}\label{eq:gm:quasilocal-main}
 \begin{split}
 I(A:C\mid B)_{\rho_{\Lambda,\beta}}
 \le
 C_\beta\exp\!\bigl[
 C_\beta(1+|\partial_{\mathrm e}A|)^{\theta/D}
 -c_\beta(1+r)^{\theta/(D+1-\theta)}
 \bigr].
 \end{split}
\end{equation}
The boundary $\partial_{\mathrm e}A$ is taken in the ambient lattice $\Z^D$, including edges leaving $\Lambda$.  
\end{theorem}

\begin{theorem}[Power-law interactions of unrestricted body size]\label{thm:gm:powerlaw}
Let $\Gamma=\Z^D$ with the nearest-neighbor graph metric and uniformly bounded on-site dimensions.  Assume
\begin{equation}\label{eq:gm:powerlaw-assumption}
 J_\zeta:=\sup_{x,y\in\Z^D}
 (1+d(x,y))^\zeta
 \sum_{\substack{X\Subset\Z^D\\x,y\in X}}\|\Phi(X)\|_\infty
 <\infty,
 \qquad \zeta>D,
\end{equation}
with the supremum including $x=y$.  No bound on $|X|$ or on the number of terms meeting a site is required.  Fix $\beta>0$ and a full partition $\Lambda=A\sqcup B\sqcup C$ with nonempty $A,C$.  Write $r=d(A,C)$ and
\begin{equation}\label{eq:gm:powerlaw-exponents}
 b_\zeta=1+\frac{2\zeta}{D},
 \qquad
 \nu_\zeta=\frac{b_\zeta(\zeta-2D)}{\zeta-D+b_\zeta}.
\end{equation}
For $D<\zeta<2D$,
\begin{equation}\label{eq:gm:powerlaw-log-main}
 I(A:C\mid B)_{\rho_{\Lambda,\beta}}
 \le C_\beta(1+|A|)^{b_\zeta+1}[\log(e+r)]^{-b_\zeta}.
\end{equation}
At the endpoint $\zeta=2D$,
\begin{equation}\label{eq:gm:powerlaw-log-endpoint}
 I(A:C\mid B)_{\rho_{\Lambda,\beta}}
 \le C_\beta(1+|A|)^6[\log(e+r)]^{-5}
 [\log(e+\log(e+r))]^2.
\end{equation}
For $\zeta>2D$,
\begin{equation}\label{eq:gm:powerlaw-main}
 I(A:C\mid B)_{\rho_{\Lambda,\beta}}
 \le C_\beta(1+|A|)^{b_\zeta+1}(1+r)^{-\nu_\zeta}.
\end{equation}
The constants depend only on $\beta,D,J_\zeta,\zeta$ and the fixed on-site dimension bound, not on $\Lambda,A,B,C$ or $r$.  
\end{theorem}

Unless indicated otherwise, constants denoted by $C,c,C_\beta,c_\beta$ may be enlarged or decreased from one occurrence to the next.  
They never depend on region sizes. All region-size dependence is displayed explicitly. 

\section{Relative-entropy loss and its variational bound}\label{sec:gm:modular}

For faithful states $\rho,\sigma$ on a tensor product $R\otimes E$, put
\begin{equation}\label{eq:gm:loss}
 \delta_R(\rho,\sigma)
 =D(\rho\|\sigma)-D(\rho_R\|\sigma_R).
\end{equation}

\subsection{Comparison with a cut Gibbs Hamiltonian}

\begin{lemma}[Cut-Hamiltonian comparison]\label{lem:gm:cut}
Let $H_\Lambda=H_0+V$, where $V$ is given by \eqref{eq:gm:cut-strength}.  For $R=A,BC$, define
\begin{equation}\label{eq:gm:regional-gibbs-states}
 H_R=\sum_{X\subseteq R}\Phi(X),\qquad
 \gamma_R=\frac{e^{-\beta H_R}}{\Tr e^{-\beta H_R}},
 \qquad \gamma_B=\Tr_C\gamma_{BC}.
\end{equation}
Set $Z_0=\Tr e^{-\beta H_0}$ and $Z=\Tr e^{-\beta H_\Lambda}=Z_\Lambda$.  Then
\begin{equation}\label{eq:gm:cut-reference-pair}
 H_0=H_A+H_{BC},\qquad
 \sigma=\frac{e^{-\beta H_0}}{Z_0}=\sigma_A\otimes\sigma_{BC},\qquad
 \rho=\frac{e^{-\beta(H_0+V)}}Z
\end{equation}
satisfy
\begin{equation}\label{eq:gm:comparison}
 \delta_{AB}(\rho,\sigma)
 =I(A:C\mid B)_\rho
   +D(\rho_{BC}\|\gamma_{BC})-D(\rho_B\|\gamma_B), \quad \text{and} \quad
 I(A:C\mid B)_\rho\le\delta_{AB}(\rho,\sigma).
\end{equation}

Moreover,
\begin{equation}\label{eq:gm:free-energy}
 |\log Z-\log Z_0|\le\beta\|V\|_\infty\le\beta g_{A\mid A^c}.
\end{equation}

\end{lemma}

\begin{proof}
Expand the two relative entropies, use the tensor-product form of $\sigma$, and collect the four entropies appearing in CMI.  This gives the first line of \eqref{eq:gm:comparison}.  Data processing under $\Tr_C$ makes its last difference nonnegative.

For the free-energy estimate, let $Z(s)=\Tr e^{-\beta(H_0+sV)}$.  Differentiation under the trace gives
\begin{equation}\label{eq:gm:partition-function-derivative}
 \frac{d}{ds}\log Z(s)=-\beta\Tr(\rho_s V),
 \qquad \rho_s=Z(s)^{-1}e^{-\beta(H_0+sV)}.
\end{equation}
Its absolute value is at most $\beta\|V\|_\infty$.  Integration from $0$ to $1$ proves \eqref{eq:gm:free-energy}.
\end{proof}

\subsection{A variational resolvent formula}
For the moment, $\rho,\sigma$ are arbitrary faithful states on $R\otimes E$.  Let
\begin{equation}\label{eq:gm:relative-modular}
 \Delta=L_\sigma R_{\rho^{-1}},\qquad
 \Delta(Y)=\sigma Y\rho^{-1},\qquad
 K=\log\Delta.
\end{equation}

Here $L_X$ and $R_X$ denote left and right multiplication on Hilbert--Schmidt space, respectively: $L_X(Y)=XY$ and $R_X(Y)=YX$.  The operators $L_\sigma$ and $R_{\rho^{-1}}$ are self-adjoint and commute, so $\Delta$ is self-adjoint.  Faithfulness of $\rho$ and $\sigma$ also gives, for every nonzero $X$,
\begin{equation}\label{eq:gm:modular-positivity}
 \langle X,\Delta(X)\rangle
 =\Tr(X^\dagger\sigma X\rho^{-1})
 =\|\sigma^{1/2}X\rho^{-1/2}\|_2^2>0.
\end{equation}
Thus $\Delta$ is strictly positive on Hilbert--Schmidt space, and $K=\log\Delta$ is well defined and self-adjoint.  If $K=\sum_j k_j P_j$ is its spectral decomposition, where $P_j$ projects onto the eigenspace with eigenvalue $k_j$, then $f(K)=\sum_j f(k_j)P_j$.  In particular, $\one_{\{|K|\ge\ell\}}=\sum_{j:\,|k_j|\ge\ell}P_j$ projects onto the eigenspaces whose eigenvalues have absolute value at least $\ell$.

\begin{lemma}[Resolvent representation of the entropy loss]\label{lem:gm:gap-integral}
Let $\rho,\sigma$ be faithful states on the finite-dimensional tensor product $R\otimes E$, and set $\Delta_R=L_{\sigma_R}R_{\rho_R^{-1}}$.  With $\Delta$ as in \eqref{eq:gm:relative-modular},
\begin{equation}\label{eq:gm:gap-integral}
 \begin{split}
 \delta_R(\rho,\sigma)&=\int_0^\infty j(s)\,ds,\\
 j(s)&=\langle\rho^{1/2},(\Delta+s)^{-1}\rho^{1/2}\rangle
 -\langle\rho_R^{1/2},(\Delta_R+s)^{-1}\rho_R^{1/2}\rangle.
 \end{split}
\end{equation}
\end{lemma}

The lemma can be proved directly using standard resolvent formula for $\log(x)$. The proof is given in Appendix~\ref{app:gm:gap-integral}.

\begin{lemma}[Variational representation of $j(s)$]\label{lem:gm:variational}
Under the setup of Lemma~\ref{lem:gm:gap-integral}, for every $s>0$,
\begin{equation}\label{eq:gm:variational}
 j(s)=\inf_{O_R\in\mathcal B(\mathcal H_R)}
 \bigl\|(\Delta+s)^{1/2}\bigl[(O_R\otimes\one_E)
 -(\Delta+s)^{-1}\bigr]\rho^{1/2}\bigr\|_2^2.
\end{equation}
Here $O_R\otimes\one_E$ acts by left multiplication on Hilbert--Schmidt space.
The infimum is over all operators on $R$ and is attained uniquely at
$O_R=\bigl[(\Delta_R+s)^{-1}\rho_R^{1/2}\bigr]\rho_R^{-1/2}$.
\end{lemma}

Equation~\eqref{eq:gm:variational} is an operator form of the resolvent-gap identity in \cite[Lemma~2.1 and Eqs.~(2.8)--(2.9)]{CV}.

\begin{proof}
Expanding the squared norm and taking partial traces gives
\begin{equation}
\begin{aligned}
 &\bigl\|(\Delta+s)^{1/2}\bigl[(O_R\otimes\one_E)\rho^{1/2}
 -(\Delta+s)^{-1}\rho^{1/2}\bigr]\bigr\|_2^2\\
 &=\Tr(O_R^\dagger\sigma_R O_R)+s\Tr(\rho_R O_R^\dagger O_R)
 -2\operatorname{Re}\Tr(\rho_R O_R)
 +\langle\rho^{1/2},(\Delta+s)^{-1}\rho^{1/2}\rangle.
\end{aligned}
\end{equation}
Set $Y=O_R\rho_R^{1/2}$. Faithfulness makes this a bijective change of variables on the reduced Hilbert--Schmidt space. Completing the square yields
\begin{equation}\label{eq:gm:resolvent-variational-square-completion}
\begin{aligned}
 &\langle Y,(\Delta_R+s)Y\rangle
 -2\operatorname{Re}\langle Y,\rho_R^{1/2}\rangle
 +\langle\rho^{1/2},(\Delta+s)^{-1}\rho^{1/2}\rangle =\bigl\|(\Delta_R+s)^{1/2}
 [Y-(\Delta_R+s)^{-1}\rho_R^{1/2}]\bigr\|_2^2+j(s).
\end{aligned}
\end{equation}
Strict positivity of $\Delta_R+s$ proves the stated minimum and its uniqueness.
Thus $j(s)\ge0$, and every trial operator gives an upper bound on $j(s)$.
\end{proof}

Define the operator conditional expectation $\mathsf E_R$, the modular leakage $\eta_R(t)$, and the integrated modular leakage $\mathfrak q_R$ by
\begin{equation}\label{eq:gm:q-def}
 \boxed{\displaystyle
 \mathsf E_R(O):=\frac{\Tr_E O}{\dim\mathcal H_E}\otimes\one_E,
 \quad
 \eta_R(t):=\|\sigma^{it}\rho^{-it}-\mathsf E_R(\sigma^{it}\rho^{-it})\|_\infty,
 \quad
 \mathfrak q_R:=\frac12\int_\R\frac{\eta_R(t)}{|\sinh(\pi t)|}\,dt.
 }
\end{equation}
This integral is finite.  Indeed, $\sigma^{it}\rho^{-it}$ equals $\one$ at $t=0$ and is differentiable there, and $\eta_R(t)\le2$ for all $t$.

\begin{lemma}[Variational localization of the entropy loss]\label{lem:gm:resolvent}
For every $\ell>0$,
\begin{equation}\label{eq:gm:cutoff-inequality}
 \displaystyle
 \delta_R(\rho,\sigma)
 \le \langle\rho^{1/2},\log(1+e^{-\ell-K})\rho^{1/2}\rangle
       +e^{-\ell}
       +\mathfrak q_R^2(e^\ell+2\ell).
\end{equation}

The first two terms bound the small- and large-resolvent contributions; the last term measures spatial leakage. Increasing $\ell$ reduces the two resolvent-tail bounds but amplifies the leakage term. The subsequent proofs choose $\ell$ by balancing these competing effects.

More generally, for $0<a<1<b$ and $0<\alpha\le1$,
\begin{equation}\label{eq:gm:two-cutoffs}
 \delta_R(\rho,\sigma)
 \le \frac{\Tr(\rho^{1+\alpha}\sigma^{-\alpha})}{\alpha}a^\alpha
       +b^{-1}
       +\mathfrak q_R^2\left(a^{-1}+\log\frac ba\right).
\end{equation}

\end{lemma}

\begin{proof}
By Lemma~\ref{lem:gm:gap-integral}, $\delta_R(\rho,\sigma)=\int_0^\infty j(s)\,ds$.  We bound $j(s)$ from above using the variational representation in Lemma~\ref{lem:gm:variational}
with the following scalar identity (derived in Appendix~\ref{app:gm:fourier-resolvent})
\begin{equation}\label{eq:gm:fourier-resolvent}
 \frac1{s+e^u}
 =\frac1{1+s}
  +\frac{i}{2s}\int_\R
   \frac{s^{-it}(e^{itu}-1)}{\sinh(\pi t)}\,dt
 \qquad(s>0,\ u\in\R)
\end{equation}
First, with the fact that $\Delta^{it}\rho^{1/2}=\sigma^{it}\rho^{-it}\rho^{1/2}$, we have
\begin{equation}\label{eq:gm:resolvent-cocycle-operator}
 (\Delta+s)^{-1}\rho^{1/2}=B_s\rho^{1/2},\qquad
 B_s=\frac{\one}{1+s}
 +\frac{i}{2s}\int_\R
   \frac{s^{-it}(\sigma^{it}\rho^{-it}-\one)}{\sinh(\pi t)}\,dt.
\end{equation}
Since $\mathsf E_R(B_s)$ is supported on $R$, choosing
$O_R=\operatorname{Tr}_E(B_s)/\dim\mathcal H_E$ gives an admissible trial operator in \eqref{eq:gm:variational}.
Moreover, plugging the above expression of $B_s$, we have
\begin{equation}\label{eq:gm:Bs-localization}
\|B_s-\mathsf E_R(B_s)\|_\infty\le\frac{\mathfrak q_R}{s}.
\end{equation}

For any operator $O$,
\begin{equation}\label{eq:gm:energy-weight}
 \begin{split}
 \Delta^{1/2}(O\rho^{1/2})&=\sigma^{1/2}O,\\
 \|(\Delta+s)^{1/2}(O\rho^{1/2})\|_2^2
 &=\|\sigma^{1/2}O\|_2^2+s\|O\rho^{1/2}\|_2^2
 \le(1+s)\|O\|_\infty^2.
 \end{split}
\end{equation}

Substituting $O_R=\operatorname{Tr}_E(B_s)/\dim\mathcal H_E$ into \eqref{eq:gm:variational} and using $(\Delta+s)^{-1}\rho^{1/2}=B_s\rho^{1/2}$ gives
\begin{equation}\label{eq:gm:j-bound}
 \begin{aligned}
 0\le j(s)
 &\le\bigl\|(\Delta+s)^{1/2}\bigl(\mathsf E_R(B_s)\rho^{1/2}-(\Delta+s)^{-1}\rho^{1/2}\bigr)\bigr\|_2^2\\
 &=\bigl\|(\Delta+s)^{1/2}\bigl[(\mathsf E_R(B_s)-B_s)\rho^{1/2}\bigr]\bigr\|_2^2\\
 &=\bigl\|\sigma^{1/2}(\mathsf E_R(B_s)-B_s)\bigr\|_2^2
   +s\bigl\|(\mathsf E_R(B_s)-B_s)\rho^{1/2}\bigr\|_2^2\\
 &\le(1+s)\|\mathsf E_R(B_s)-B_s\|_\infty^2
 \le\frac{1+s}{s^2}\mathfrak q_R^2.
 \end{aligned}
\end{equation}
Here the second equality uses \eqref{eq:gm:energy-weight}; the last line uses $\Tr\sigma=\Tr\rho=1$ and \eqref{eq:gm:Bs-localization}.

For the small-resolvent interval, positivity of the reduced resolvent gives
\begin{equation}\label{eq:gm:small-tail}
 \int_0^a j(s)\,ds
 \le\langle\rho^{1/2},\log(1+a\Delta^{-1})\rho^{1/2}\rangle.
\end{equation}

For the large-resolvent interval, $\langle\rho^{1/2},(\Delta+s)^{-1}\rho^{1/2}\rangle\le1/s$.  Writing
\begin{equation}\label{eq:gm:reduced-modular-spectral-weights}
 \Delta_R=\sum_j\lambda_jP_j,\qquad
 p_j=\|P_j(\rho_R^{1/2})\|_2^2\ge0,\qquad
 \sum_jp_j=\Tr\rho_R=1,
\end{equation}
we have
\begin{equation}\label{eq:gm:reduced-resolvent-jensen}
 \begin{aligned}
 \sum_jp_j\lambda_j
 &=\langle\rho_R^{1/2},\Delta_R\rho_R^{1/2}\rangle
 =\Tr(\rho_R^{1/2}\sigma_R\rho_R^{-1/2})
 =\Tr\sigma_R=1,\\
 \langle\rho_R^{1/2},(\Delta_R+s)^{-1}\rho_R^{1/2}\rangle
 &=\sum_j\frac{p_j}{\lambda_j+s}
 \underset{\text{Jensen}}{\ge}\frac{1}{\sum_jp_j\lambda_j+s}
 =\frac1{1+s}.
 \end{aligned}
\end{equation}
Jensen's inequality can be applied because $\left(\frac1{x+s}\right)''=\frac{2}{(x+s)^3}>0.$ is convex.
Consequently,
\begin{equation}\label{eq:gm:large-tail}
 \int_b^\infty j(s)\,ds \le\int_b^\infty\left(\frac1s-\frac1{s+1}\right)\,ds \le\log(1+b^{-1})\le b^{-1}.
\end{equation}

Integrating \eqref{eq:gm:j-bound} over $[a,b]$ and choosing
$a=e^{-\ell}$, $b=e^\ell$ proves \eqref{eq:gm:cutoff-inequality}.
Finally, $\log(1+w)\le w^\alpha/\alpha$ for $w\ge0$ and $0<\alpha\le1$, while
\begin{equation}\label{eq:gm:negative-modular-moment}
 \langle\rho^{1/2},\Delta^{-\alpha}\rho^{1/2}\rangle
 =\Tr(\rho^{1+\alpha}\sigma^{-\alpha}).
\end{equation}
These observations in \eqref{eq:gm:small-tail} prove \eqref{eq:gm:two-cutoffs}.
\end{proof}

\begin{corollary}[Optimized negative-moment bound]\label{cor:gm:moment}
For $0<\alpha\le1$,
\begin{equation}\label{eq:gm:optimized-moment}
 \displaystyle
 \delta_R(\rho,\sigma)
 \le\left(\frac1\alpha+3\right)
 \bigl[\Tr(\rho^{1+\alpha}\sigma^{-\alpha})\bigr]^{1/(1+\alpha)}
 \mathfrak q_R^{\,2\alpha/(1+\alpha)}.
\end{equation}

\end{corollary}

\begin{proof}
Jensen's inequality gives $\Tr(\rho^{1+\alpha}\sigma^{-\alpha})\ge e^{\alpha D(\rho\|\sigma)}\ge1$.
If $0<\mathfrak q_R<1$, put
\begin{equation}\label{eq:gm:optimized-resolvent-cutoffs}
 a=\left(\frac{\mathfrak q_R^2}{\Tr(\rho^{1+\alpha}\sigma^{-\alpha})}\right)^{1/(1+\alpha)},
 \qquad b=a^{-1},
 \qquad B=\bigl[\Tr(\rho^{1+\alpha}\sigma^{-\alpha})\bigr]^{1/(1+\alpha)}\mathfrak q_R^{2\alpha/(1+\alpha)}.
\end{equation}
Then $\Tr(\rho^{1+\alpha}\sigma^{-\alpha}) a^\alpha=B$, $\mathfrak q_R^2/a=B$, and $a\le B$.
Also $2\mathfrak q_R^2\log(1/a)=2B a\log(1/a)\le2B/e$.
Substitution in \eqref{eq:gm:two-cutoffs} proves the assertion, since $2+2/e<3$.
If $\mathfrak q_R=0$, \eqref{eq:gm:j-bound} gives $\delta_R=0$.
If $\mathfrak q_R\ge1$, use
\begin{equation}\label{eq:gm:large-leakage-entropy-bound}
 \delta_R\le D(\rho\|\sigma)\le\frac{\log \Tr(\rho^{1+\alpha}\sigma^{-\alpha})}{\alpha}
 \le\frac{1+\alpha}{\alpha}\bigl[\Tr(\rho^{1+\alpha}\sigma^{-\alpha})\bigr]^{1/(1+\alpha)},
\end{equation}
which is again bounded by the right-hand side of \eqref{eq:gm:optimized-moment}.
\end{proof}

\section{Finite-range proof and optimization of the cut dependence}\label{sec:gm:local-proof}

From now on $\rho,\sigma,H_0,V$ are the cut Gibbs pair of Lemma~\ref{lem:gm:cut}, and the retained region is $AB$.

\subsection{Finite-range interactions}

\begin{lemma}[Finite-range integrated modular leakage]\label{lem:gm:local-leakage}
Under the hypotheses of Theorem~\ref{thm:gm:local}, there are $C_0,v,\mu_{\mathrm{LR}}>0$, independent of $\Lambda,A,B,C$, such that
\begin{equation}\label{eq:gm:local-q}
 \mathfrak q_{AB}
 \le C_0(1+\beta g_{A\mid A^c})
       \exp\!\left[-\lambda_\beta(r-R_0)_+\right],
 \qquad
 \lambda_\beta=\frac{\pi\mu_{\mathrm{LR}}}{\pi+\beta v}.
\end{equation}

Here $(s)_+=\max\{s,0\}$.
\end{lemma}

\begin{proof}
\emph{Step 1: General leakage bounds.} Using
$\sigma^{it}\rho^{-it} = e^{-i\beta tH_0}e^{i\beta t(H_0+V)}$.
For real $u$, the interaction-picture evolution is generated by
\begin{equation}\label{eq:gm:interaction-picture}
 V(u)=e^{-iuH_0}Ve^{iuH_0},\qquad
 \frac{d}{du}\bigl(e^{-iuH_0}e^{iu(H_0+V)}\bigr)
 =iV(u)e^{-iuH_0}e^{iu(H_0+V)}.
\end{equation}

Integrating the evolution equation from its initial value $\one$ and using $\|V(u)\|_\infty\le g_{A\mid A^c}$, gives the linear bound below; unitarity and contractivity give the bound $2$.  Integration against $1/(2|\sinh(\pi t)|)$, with $\int_\R |t|/|\sinh(\pi t)|\,dt=1/2$, then yields
\begin{equation}\label{eq:gm:elementary-leakage}
 \eta_{AB}(t)\le\min\{2,2\beta g_{A\mid A^c}|t|\},
 \qquad \mathfrak q_{AB}\le\frac{\beta g_{A\mid A^c}}{2}.
\end{equation}

\emph{Step 2: Spatial decay of the leakage.} To obtain a bound sensitive to spatial locality, let $W_{AB}(u)$ solve
\begin{equation}\label{eq:gm:localized-cocycle-evolution}
 \frac{d}{du}W_{AB}(u)=i\mathsf E_{AB}(V(u))W_{AB}(u),
 \qquad W_{AB}(0)=\one.
\end{equation}
Since its generator is self-adjoint and supported on $AB$, $W_{AB}(u)$ is unitary and $\mathsf E_{AB}(W_{AB}(u))=W_{AB}(u)$.  Duhamel's formula (Lemma~\ref{lem:gm:duhamel} in Appendix~\ref{app:gm:gronwall}) and unitarity give
\begin{equation}\label{eq:gm:cocycle-duhamel-comparison}
 \bigl\|e^{-i\beta tH_0}e^{i\beta t(H_0+V)}-W_{AB}(\beta t)\bigr\|_\infty
 \le\int_0^{\beta|t|}
 \|V(\operatorname{sgn}(t)u)-\mathsf E_{AB}(V(\operatorname{sgn}(t)u))\|_\infty\,du.
\end{equation}
Since $(\mathrm{id}-\mathsf E_{AB})(W_{AB}(\beta t))=0$, subtracting $W_{AB}(\beta t)$ does not change the leakage.  Contractivity of $\mathsf E_{AB}$ therefore gives
\begin{equation}\label{eq:gm:leakage-from-cocycle-comparison}
 \begin{aligned}
 \eta_{AB}(t)
 &=\bigl\|(\mathrm{id}-\mathsf E_{AB})
 \bigl(e^{-i\beta tH_0}e^{i\beta t(H_0+V)}-W_{AB}(\beta t)\bigr)\bigr\|_\infty\\
 &\le2\bigl\|e^{-i\beta tH_0}e^{i\beta t(H_0+V)}-W_{AB}(\beta t)\bigr\|_\infty.
 \end{aligned}
\end{equation}
Combining these estimates yields
\begin{equation}\label{eq:gm:duhamel-leakage}
 \eta_{AB}(t)
 \le 2\int_0^{\beta|t|}
 \|V(\operatorname{sgn}(t)u)-\mathsf E_{AB}(V(\operatorname{sgn}(t)u))\|_\infty\,du.
\end{equation}

The finite-range Lieb--Robinson bound and local approximation (cf.~\cite[Eq.~(3.26) and Corollary~4.4]{NachtergaeleSimsYoung2019}) give, for every original interaction term,
\begin{equation}\label{eq:gm:local-LR}
 \|e^{-iuH_0}h_\gamma e^{iuH_0}
   -\mathsf E_{AB}(e^{-iuH_0}h_\gamma e^{iuH_0})\|_\infty
 \le C_{\rm LR}\|h_\gamma\|_\infty
 e^{v|u|-\mu_{\mathrm{LR}} d(\supp h_\gamma,C)}.
\end{equation}
Here $C_{\rm LR}>0$ is a prefactor, $v>0$ is the time-growth rate, and $\mu_{\mathrm{LR}}>0$ is the spatial decay rate; these constants depend only on the fixed interaction parameters and are independent of $\Lambda,A,B,C$.  The distance $d(\supp h_\gamma,C)$ is the minimum distance between the support of $h_\gamma$ and the region $C$.
For every term contributing to $V$, its support meets $A$ and has diameter at most $R_0$, whence
$d(\supp h_\gamma,C)\ge(r-R_0)_+$.  Sum \eqref{eq:gm:local-LR} over these terms and apply \eqref{eq:gm:duhamel-leakage}.  This gives
\begin{equation}\label{eq:gm:finite-range-pointwise-leakage}
 \eta_{AB}(t)
 \le\min\{2,C\beta g_{A\mid A^c}|t|
 e^{v\beta|t|-\mu_{\mathrm{LR}}(r-R_0)_+}\}.
\end{equation}
\emph{Step 3: Integration over modular time.} Set $T=\mu_{\mathrm{LR}}(r-R_0)_+/(\pi+\beta v)$.  On $|t|\le T$ use the localized bound; on $|t|>T$ use $\eta_{AB}(t)\le2$.  For the latter contribution, we use
\begin{equation}\label{eq:gm:kernel-tail}
 \int_{|t|>T}\frac{dt}{|\sinh(\pi t)|}\le C e^{-\pi T}\qquad(T\ge1).
\end{equation}
Using $\int_\R |t|/|\sinh(\pi t)|\,dt=1/2$ for the short-time contribution gives, for $T\ge1$,
\begin{equation}\label{eq:gm:finite-range-time-split}
 \begin{aligned}
 \mathfrak q_{AB}
 &\le C\beta g_{A\mid A^c}
 e^{v\beta T-\mu_{\mathrm{LR}}(r-R_0)_+}+Ce^{-\pi T}\\
 &\le C(1+\beta g_{A\mid A^c})e^{-\pi T},
 \end{aligned}
\end{equation}
since the choice of $T$ makes $v\beta T-\mu_{\mathrm{LR}}(r-R_0)_+=-\pi T$.
If $0\le T<1$, the elementary bound \eqref{eq:gm:elementary-leakage} gives
\begin{equation}\label{eq:gm:finite-range-short-distance-leakage}
 \mathfrak q_{AB}\le\frac{\beta g_{A\mid A^c}}2
 \le\frac{e^\pi}{2}(1+\beta g_{A\mid A^c})e^{-\pi T},
\end{equation}
since $e^{\pi(1-T)}\ge1$.  Thus choosing $C_0\ge\max\{C,e^\pi/2\}$ covers both $T\ge1$ and $T<1$.  Finally, $\pi T=\lambda_\beta(r-R_0)_+$, which proves \eqref{eq:gm:local-q}.
\end{proof}

\subsection{Moment bound and cut optimization}

\begin{lemma}[Short-imaginary-time moment]\label{lem:gm:imaginary}
Under the finite-range hypotheses, for $0<\alpha\le1$ with
$2\alpha\beta\mathfrak dJ<1$, the cut Gibbs pair satisfies
\begin{equation}\label{eq:gm:sharp-moment}
 \Tr(\rho^{1+\alpha}\sigma^{-\alpha})
 \le
 \exp\!\left[
 \alpha\beta g_{A\mid A^c}
 -\frac{g_{A\mid A^c}}{2\mathfrak dJ}\log(1-2\alpha\beta\mathfrak dJ)
 \right].
\end{equation}

In particular, if
\begin{equation}\label{eq:gm:alpha0}
 \alpha_0=\min\left\{1,\frac1{4\beta\mathfrak dJ}\right\},
\end{equation}

then for $0<\alpha\le\alpha_0$,
\begin{equation}\label{eq:gm:coarse-moment}
 \Tr(\rho^{1+\alpha}\sigma^{-\alpha})\le e^{3\alpha\beta g_{A\mid A^c}}.
\end{equation}

\end{lemma}

\begin{proof}
We use the nested-commutator counting argument of \cite[Lemma~IX.3 and its proof]{Chen2025}, applied here to an original interaction term $h_\gamma$ with the interaction norm bound $J$ restored.  Writing $\ad_{H_0}(O)=[H_0,O]$ and $\ad_{H_0}^{\,n}$ for its $n$-fold iteration, it gives
\begin{equation}\label{eq:gm:local-commutators}
 \|\ad_{H_0}^{\,n}(h_\gamma)\|_\infty
 \le n!(2\mathfrak dJ)^n\|h_\gamma\|_\infty.
\end{equation}

After $m$ commutators, a nonzero contribution is supported in a union of at most $m+1$ original interaction supports.  At most $\mathfrak d(m+1)$ terms can contribute to the next commutator, each at cost at most $2J$.  Multiplication of these bounds proves \eqref{eq:gm:local-commutators}.  Summing the resulting geometric series, as in \cite[Lemma~IX.3]{Chen2025}, and then summing over the terms of $V$ gives, for $2\mathfrak dJ|u|<1$,
\begin{equation}\label{eq:gm:imaginary-V}
 \|e^{-uH_0}Ve^{uH_0}\|_\infty
 \le\frac{g_{A\mid A^c}}{1-2\mathfrak dJ|u|}.
\end{equation}
Let $Q(u)=e^{-u(H_0+V)}e^{uH_0}$.  It satisfies
\begin{equation}\label{eq:gm:imaginary-cocycle-evolution}
 Q'(u)=-Q(u)e^{-uH_0}Ve^{uH_0},\qquad Q(0)=\one.
\end{equation}
Integrating this equation, taking operator norms, and using \eqref{eq:gm:imaginary-V} gives, for $0\le u<(2\mathfrak dJ)^{-1}$,
\begin{equation}\label{eq:gm:imaginary-cocycle-integral-bound}
 \|Q(u)\|_\infty
 \le 1+\int_0^u\frac{g_{A\mid A^c}}{1-2\mathfrak dJs}\|Q(s)\|_\infty\,ds.
\end{equation}
Applying Gr\"onwall's inequality (Lemma~\ref{lem:gm:gronwall} in Appendix~\ref{app:gm:gronwall}) with $f(u)=\|Q(u)\|_\infty$, $c=1$, and $a(s)=g_{A\mid A^c}/(1-2\mathfrak dJs)$ gives
\begin{equation}\label{eq:gm:imaginary-Q}
 \|Q(u)\|_\infty
 \le\exp\left(\int_0^u\frac{g_{A\mid A^c}}{1-2\mathfrak dJs}\,ds\right)
 =(1-2\mathfrak dJu)^{-g_{A\mid A^c}/(2\mathfrak dJ)}.
\end{equation}

Since
\begin{equation}\label{eq:gm:gibbs-moment-cocycle-identity}
 \Tr(\rho^{1+\alpha}\sigma^{-\alpha})
 =\left(\frac{Z_0}{Z}\right)^\alpha
 \Tr\left[\rho\,e^{-\alpha\beta(H_0+V)}e^{\alpha\beta H_0}\right],
\end{equation}
\eqref{eq:gm:free-energy} and \eqref{eq:gm:imaginary-Q} prove \eqref{eq:gm:sharp-moment}.  The trace on the left is positive; its absolute-value estimate by the operator norm of $Q$ is therefore sufficient even though $Q$ need not be self-adjoint.
For $\alpha\le\alpha_0$, the denominator in \eqref{eq:gm:imaginary-V} is at least $1/2$ on $[0,\alpha\beta]$.  Thus $\|Q(\alpha\beta)\|_\infty\le e^{2\alpha\beta g_{A\mid A^c}}$, which proves \eqref{eq:gm:coarse-moment}.
\end{proof}

\begin{proof}[Proof of Theorem~\ref{thm:gm:local}]
Combine Lemma~\ref{lem:gm:cut}, Corollary~\ref{cor:gm:moment}, Lemma~\ref{lem:gm:local-leakage}, and \eqref{eq:gm:sharp-moment}.  For every admissible $\alpha$,
\begin{equation}\label{eq:gm:local-optimized}
 \begin{split}
 I(A:C\mid B)_\rho
 &\le\left(\frac1\alpha+3\right)
       [C_0(1+\beta g_{A\mid A^c})]^{\frac{2\alpha}{1+\alpha}}\exp\left\{
 \frac{g_{A\mid A^c}}{1+\alpha}
 \left[\alpha\beta-\frac{\log(1-2\alpha\beta\mathfrak dJ)}{2\mathfrak dJ}\right]
 -\frac{2\alpha\lambda_\beta}{1+\alpha}(r-R_0)_+
 \right\}.
 \end{split}
\end{equation}

The infimum over $0<\alpha\le1$ with $2\alpha\beta\mathfrak dJ<1$ may be taken.

Choose $\alpha=\alpha_0 =\min\left\{1,\frac1{4\beta\mathfrak dJ}\right\}$. The polynomial factor in $g_{A\mid A^c}$ can be absorbed into $C_\beta e^{C_\beta g_{A\mid A^c}}$, and the fixed range shift into $C_\beta$.  This proves \eqref{eq:gm:local-main}.  One explicit decay length before these absorptions is
\begin{equation}\label{eq:gm:local-length}
 \xi_\beta=\frac{1+\alpha_0}{2\alpha_0\lambda_\beta}
 =\frac{1+\alpha_0}{2\alpha_0}\frac{\pi+\beta v}{\pi\mu_{\mathrm{LR}}}.
\end{equation}

For a fixed finite $A$, $g_{A\mid A^c}\le\mathfrak dJ|A|$: all terms containing a given site overlap, so there are at most $\mathfrak d$ of them.  The bound is consequently uniform in $\Lambda$ and $C$ and tends to zero with $r$.

On $\Z^D$, every cut term contains an $A$-site within distance $R_0$ of a nearest-neighbor boundary edge.  The number of such sites is at most $C_{D,R_0}|\partial_{\mathrm e}A|$.  Summing the at-most-$\mathfrak dJ$ interaction strength at each of these sites gives $g_{A\mid A^c}\le C_{D,R_0}\mathfrak dJ|\partial_{\mathrm e}A|$, proving \eqref{eq:gm:local-boundary}.
In summary, the exponential distance decay comes from the Lieb--Robinson bound, while the short-imaginary-time moment bound produces the exponential dependence on the cut strength and hence on the boundary size.
\end{proof}

\section{Quasi-local interactions: spatial localization and modular tails}\label{sec:gm:quasilocal-proof}

Compared with the finite-range case, two additional difficulties arise: individual interactions can span the shielding region, and their supports can contain arbitrarily many sites.  The bounded-degree commutator counting used in Lemma~\ref{lem:gm:imaginary} therefore no longer applies directly, so we cannot rely on the same short-imaginary-time bound for the negative modular moment.
We therefore apply the bound \eqref{eq:gm:cutoff-inequality} in Lemma~\ref{lem:gm:resolvent} directly, rather than Corollary~\ref{cor:gm:moment}, and control the small-resolvent contribution directly through spectral-tail estimates for $K=\log\Delta$, the logarithm of the relative modular operator defined in \eqref{eq:gm:relative-modular}.

We first remove interactions of large diameter that touch a neighborhood of $A$.  A Gibbs-interpolation estimate bounds the resulting change in CMI, while the retained interactions admit a spatial localization estimate inside that neighborhood.  To obtain the required spectral-tail estimates, spatial decay bounds the strength of interactions with large supports, which in turn controls moments of the relative modular logarithm with boundary-dependent prefactors.  Combining this spectral control with spatial localization and balancing the truncation error against the spectral-tail contribution gives the stretched-exponential distance decay in Theorem~\ref{thm:gm:quasilocal}.

\subsection{Consequences of spatial decay for support cardinalities}

\begin{lemma}[Weighted support and cut estimates]\label{lem:gm:support-moments}
Under \eqref{eq:gm:F-assumption}, there exists $a_{\mathrm{supp}}>0$, depending only on $D,\mu,\theta$, such that
\begin{equation}\label{eq:gm:Ja}
 \sup_{x\in\Z^D}
 \sum_{X\ni x}e^{2a_{\mathrm{supp}}|X|^{\theta/D}}\|\Phi(X)\|_\infty
 \le C J_{\mu,\theta}<\infty.
\end{equation}

Moreover,
\begin{equation}\label{eq:gm:weighted-cut}
 \sum_{X\text{ crossing }A\mid A^c}
 e^{a_{\mathrm{supp}}|X|^{\theta/D}}\|\Phi(X)\|_\infty\le C |\partial_{\mathrm e}A|,
\end{equation}

where $C$ depends only on $D,\mu,\theta,J_{\mu,\theta}$.
\end{lemma}

\begin{proof}
For $|X|\ge2$, put $R_X=\diam X\ge1$.  A lattice counting bound gives
$|X|\le(2R_X+1)^D\le3^D R_X^D$, and therefore $|X|^{\theta/D}\le3^\theta R_X^\theta$.
For each fixed $x\in X$, some $y\in X$ obeys $d(x,y)\ge R_X/2$.  Choose
\[
 0<a_{\mathrm{supp}}\le\min\{1,\mu/(4\cdot6^\theta)\}.
\]
Then $2a_{\mathrm{supp}}|X|^{\theta/D}\le(\mu/2)d(x,y)^\theta$ for at least one such $y$, and hence
\begin{equation}\label{eq:gm:support-weight-distance-comparison}
 e^{2a_{\mathrm{supp}}|X|^{\theta/D}}\le\sum_{y\in X\setminus\{x\}}
                  e^{(\mu/2)d(x,y)^\theta}\qquad(|X|\ge2).
\end{equation}
Separating the singleton $X=\{x\}$ and interchanging the remaining nonnegative sums gives
\begin{equation}\label{eq:gm:weighted-per-site-summability}
 \begin{aligned}
 \sup_{x\in\Z^D}\sum_{X\ni x}e^{2a_{\mathrm{supp}}|X|^{\theta/D}}\|\Phi(X)\|_\infty
 &\le e^{2a_{\mathrm{supp}}}J_{\mu,\theta}
 +\sup_x\sum_{y\ne x}e^{(\mu/2)d(x,y)^\theta}
       \sum_{X\ni x,y}\|\Phi(X)\|_\infty\\
 &\le e^{2a_{\mathrm{supp}}}J_{\mu,\theta}
 +J_{\mu,\theta}\sum_{z\ne0}
      e^{-(\mu/2)|z|_1^\theta}.
 \end{aligned}
\end{equation}
Here both bounds use \eqref{eq:gm:F-assumption}, with $x=y$ for the singleton and $x\ne y$ for the pair sum; in the last line, $z=y-x$.
The lattice sum is finite for every $\mu>0$ and $\theta>0$.  This proves \eqref{eq:gm:Ja}, including singleton supports.

For $x\ne y$, Cauchy--Schwarz for the sum over interaction supports gives
\begin{equation}\label{eq:gm:weighted-pair}
 \begin{split}
 \sum_{X\ni x,y}e^{a_{\mathrm{supp}}|X|^{\theta/D}}\|\Phi(X)\|_\infty
 &\le
 \left(\sum_{X\ni x,y}\|\Phi(X)\|_\infty\right)^{1/2}
 \left(\sum_{X\ni x,y}e^{2a_{\mathrm{supp}}|X|^{\theta/D}}\|\Phi(X)\|_\infty\right)^{1/2}\\
 &\le C J_{\mu,\theta}\,e^{-(\mu/2)d(x,y)^\theta}.
 \end{split}
\end{equation}

For every displacement $z\in\Z^D$,
\begin{equation}\label{eq:gm:translation-count}
 |\{x\in A:x+z\notin A\}|\le |\partial_{\mathrm e}A||z|_1.
\end{equation}
Follow a fixed nearest-neighbor path realizing $z$.
Every translated path from $A$ to $\Z^D\setminus A$ crosses an ambient boundary edge, and for each step and boundary edge there is at most one starting point.
This proves \eqref{eq:gm:translation-count}.

Every cut support contains at least one pair $x\in A,y\notin A$, so summing over all such pairs only overcounts its contribution.  Thus \eqref{eq:gm:weighted-pair} and \eqref{eq:gm:translation-count} give
\begin{equation}\label{eq:gm:weighted-cut-pair-counting}
 \begin{aligned}
 \sum_{X\text{ crossing }A\mid A^c}e^{a_{\mathrm{supp}}|X|^{\theta/D}}\|\Phi(X)\|_\infty
 &\le\sum_{x\in A}\sum_{y\in\Z^D\setminus A}
       \sum_{X\ni x,y}e^{a_{\mathrm{supp}}|X|^{\theta/D}}\|\Phi(X)\|_\infty\\
 &\le C J_{\mu,\theta}
       \sum_{z\in\Z^D}|\{x\in A:x+z\notin A\}|e^{-(\mu/2)|z|_1^\theta}\\
 &\le C J_{\mu,\theta}\,|\partial_{\mathrm e}A|
       \sum_{z\in\Z^D}|z|_1e^{-(\mu/2)|z|_1^\theta}.
 \end{aligned}
\end{equation}
Here $z=y-x$ groups pairs by their displacement.  Since
$\sum_z|z|_1 e^{-(\mu/2)|z|_1^\theta}<\infty$, this proves \eqref{eq:gm:weighted-cut}.
\end{proof}

\subsection{Moments of the relative modular logarithm}
Recall $K=\log\Delta$ from \eqref{eq:gm:relative-modular}.  Since left and right multiplication commute,
\begin{equation}\label{eq:gm:modular-logarithm-expansion}
 K=L_{\log\sigma}-R_{\log\rho},\qquad
 \log\sigma=-\beta H_0-(\log Z_0)\one,\qquad
 \log\rho=-\beta(H_0+V)-(\log Z)\one.
\end{equation}
Using $[\rho^{1/2},\log\rho]=0$, we obtain, for every operator $O$,
\begin{equation}\label{eq:gm:K-action}
 \begin{aligned}
 K(O\rho^{1/2})
 &= (\log\sigma)O\rho^{1/2}-O\rho^{1/2}\log\rho\\
 &=\bigl((\log\sigma)O-O\log\rho\bigr)\rho^{1/2}\\
 &=\left(-\beta H_0O+\beta O(H_0+V)+(\log Z-\log Z_0)O\right)\rho^{1/2}\\
 &=\left(\beta[O,H_0]+\beta OV+(\log Z-\log Z_0)O\right)\rho^{1/2}.
 \end{aligned}
\end{equation}

Define the operator map and its iterates by
\begin{equation}\label{eq:gm:powerlaw-D-map}
 \mathcal D(O)=\beta[O,H_0]+\beta OV+(\log Z-\log Z_0)O,
\end{equation}
\begin{equation}\label{eq:gm:modular-iterates-definition}
 O_0=\one,\qquad O_j=\mathcal D(O_{j-1})=\mathcal D^j(\one).
\end{equation}
Equation~\eqref{eq:gm:K-action} gives $K^j\rho^{1/2}=O_j\rho^{1/2}$. We will estimate these vectors through explicit local decompositions $O_j=\sum_XO_{j,X}$.

\begin{lemma}[Uniform modular moments and spectral tails]\label{lem:gm:gevrey}
Under \eqref{eq:gm:F-assumption}, there is a constant $C_\beta>0$, depending only on $\beta,D,\mu,\theta,J_{\mu,\theta}$, such that for every integer $n\ge1$,
\begin{equation}\label{eq:gm:modular-moments}
 \|K^n\rho^{1/2}\|_2
 \le\left[C_\beta\left(n^{D/\theta}+|\partial_{\mathrm e}A|\right)\right]^n.
\end{equation}

Consequently, for every $\ell\ge0$,
\begin{equation}\label{eq:gm:modular-tail}
 \langle\rho^{1/2},\one_{\{|K|\ge\ell\}}\rho^{1/2}\rangle
 \le C_\beta
 \exp\!\left[C_\beta(1+|\partial_{\mathrm e}A|)^{\theta/D}-c_\beta\ell^{\theta/D}\right].
\end{equation}

\end{lemma}

\begin{proof}
Fix an integer $n\ge1$ and decrease the support weight in equal steps:
\begin{equation}\label{eq:gm:explicit-support-weights}
 s_j=a_{\mathrm{supp}}(1-j/n),\qquad j=0,\ldots,n.
\end{equation}
Starting from $O_{0,\varnothing}=\one$ and $O_{0,X}=0$ for $X\ne\varnothing$, construct local decompositions $O_j=\sum_XO_{j,X}$ recursively from
\begin{equation}\label{eq:gm:explicit-modular-recursion}
 O_j=\beta[O_{j-1},H_0]+\beta O_{j-1}V+(\log Z-\log Z_0)O_{j-1}.
\end{equation}
Specifically, assign $\beta[O_{j-1,X},\Phi(Y)]$ to $X\cup Y$ for terms $Y$ in $H_0$, and $\beta O_{j-1,X}\Phi(Y)$ to $X\cup Y$ for terms $Y$ in $V$; assign the term $(\log Z-\log Z_0)O_{j-1,X}$ to $X$. Terms with the same output support are combined using the triangle inequality. Equation~\eqref{eq:gm:K-action} then gives $K^j\rho^{1/2}=O_j\rho^{1/2}$.

For the commutator contribution, only $Y$ meeting $X$ contribute. The support inequality $|X\cup Y|^{\theta/D}\le|X|^{\theta/D}+|Y|^{\theta/D}$ and \eqref{eq:gm:Ja} give
\begin{equation}\label{eq:gm:explicit-weighted-overlap}
 \sum_{Y:Y\cap X\ne\varnothing}e^{s_j|Y|^{\theta/D}}\|\Phi(Y)\|_\infty
 \le\sum_{x\in X}\sum_{Y\ni x}e^{2a_{\mathrm{supp}}|Y|^{\theta/D}}\|\Phi(Y)\|_\infty
 \le C J_{\mu,\theta}|X|.
\end{equation}
Deleting the cut terms to form $H_0$ only decreases this sum. The loss of weight absorbs the factor $|X|$, since
\begin{equation}\label{eq:gm:support-weight-loss-optimization}
 |X|e^{-(a_{\mathrm{supp}}/n)|X|^{\theta/D}}
 \le\sup_{u\ge0}u e^{-(a_{\mathrm{supp}}/n)u^{\theta/D}}
 =\left(\frac{Dn}{e\theta a_{\mathrm{supp}}}\right)^{D/\theta}.
\end{equation}
Thus the weighted coefficient sum of the commutator contribution obeys
\begin{equation}\label{eq:gm:explicit-commutator-coefficients}
 \begin{aligned}
 \beta\sum_X\sum_{Y\text{ in }H_0}
 e^{s_j|X\cup Y|^{\theta/D}}\|[O_{j-1,X},\Phi(Y)]\|_\infty
 &\le C_\beta\sum_X|X|e^{s_j|X|^{\theta/D}}\|O_{j-1,X}\|_\infty\\
 &\le C_\beta n^{D/\theta}\sum_Xe^{s_{j-1}|X|^{\theta/D}}\|O_{j-1,X}\|_\infty.
 \end{aligned}
\end{equation}
For the product contribution, the same support inequality and the weighted-cut bound \eqref{eq:gm:weighted-cut} imply
\begin{equation}\label{eq:gm:explicit-product-coefficients}
 \begin{aligned}
 \beta\sum_X\sum_{Y\text{ in }V}
 e^{s_j|X\cup Y|^{\theta/D}}\|O_{j-1,X}\Phi(Y)\|_\infty
 &\le\beta\left(\sum_Xe^{s_j|X|^{\theta/D}}\|O_{j-1,X}\|_\infty\right)
 \left(\sum_{Y\text{ in }V}e^{s_j|Y|^{\theta/D}}\|\Phi(Y)\|_\infty\right)\\
 &\le C_\beta|\partial_{\mathrm e}A|
 \sum_Xe^{s_{j-1}|X|^{\theta/D}}\|O_{j-1,X}\|_\infty.
 \end{aligned}
\end{equation}
The contribution from $(\log Z-\log Z_0)O_{j-1}$ satisfies the same final bound, because $|\log Z-\log Z_0|\le\beta g_{A\mid A^c}\le C_\beta|\partial_{\mathrm e}A|$ and $s_j\le s_{j-1}$. Adding the three contributions gives
\begin{equation}\label{eq:gm:explicit-weighted-coefficient-recursion}
 \sum_Xe^{s_j|X|^{\theta/D}}\|O_{j,X}\|_\infty
 \le C_\beta\bigl(n^{D/\theta}+|\partial_{\mathrm e}A|\bigr)
 \sum_Xe^{s_{j-1}|X|^{\theta/D}}\|O_{j-1,X}\|_\infty.
\end{equation}
All constants are independent of $n,j$, the regions, and the finite volume. The initial weighted coefficient sum is one. Iterating $n$ times, using $s_n=0$ and $\|\rho^{1/2}\|_2=1$, yields
\begin{equation}\label{eq:gm:iterated-weight-loss-bound}
 \|K^n\rho^{1/2}\|_2
 =\|O_n\rho^{1/2}\|_2
 \le\sum_X\|O_{n,X}\|_\infty
 \le\left[C_\beta\bigl(n^{D/\theta}+|\partial_{\mathrm e}A|\bigr)\right]^n,
\end{equation}
which proves \eqref{eq:gm:modular-moments}.

For $\ell>0$ and every integer $n\ge1$, we have $\one_{\{|K|\ge\ell\}}\le\ell^{-2n}|K|^{2n}.$
Taking its expectation in $\rho^{1/2}$ and applying \eqref{eq:gm:modular-moments} gives
\begin{equation}\label{eq:gm:tail-moment-choice}
 \begin{aligned}
 \langle\rho^{1/2},\one_{\{|K|\ge\ell\}}\rho^{1/2}\rangle
 &\le\ell^{-2n}\langle\rho^{1/2},|K|^{2n}\rho^{1/2}\rangle\\
 &=\ell^{-2n}\|K^n\rho^{1/2}\|_2^2\\
 &\le\left[\frac{C_\beta(n^{D/\theta}+|\partial_{\mathrm e}A|)}{\ell}\right]^{2n}.
 \end{aligned}
\end{equation}

We first consider $\ell$ large enough so that the moment order can be chosen proportional to $\ell^{\theta/D}$ while keeping the ratio in \eqref{eq:gm:tail-moment-choice} below one.
Specifically, for $\ell\ge4C_\beta |\partial_{\mathrm e}A|$ and $\ell\ge4C_\beta 2^{D/\theta}$, choose
$n=\lfloor(\ell/(4C_\beta))^{\theta/D}\rfloor$.  The two contributions to the ratio in \eqref{eq:gm:tail-moment-choice} are each at most $1/4$.  The second threshold also ensures $n\ge\tfrac12(\ell/(4C_\beta))^{\theta/D}\ge1$.  Hence the tail is at most $2^{-2n}\le e^{-c_\beta\ell^{\theta/D}}$, with $c_\beta=(\log2)(4C_\beta)^{-\theta/D}$.
If either threshold fails, then
\begin{equation}\label{eq:gm:modular-tail-small-cutoff}
 \ell^{\theta/D}
 \le(4C_\beta)^{\theta/D}\max\{|\partial_{\mathrm e}A|^{\theta/D},2\}
 \le2(4C_\beta)^{\theta/D}(1+|\partial_{\mathrm e}A|)^{\theta/D}.
\end{equation}
In this remaining range, the spectral probability is at most one, so \eqref{eq:gm:modular-tail} is satisfied after enlarging the coefficient of $(1+|\partial_{\mathrm e}A|)^{\theta/D}$ to make its right-hand side at least one.
Combining the two ranges and enlarging $C_\beta$ proves \eqref{eq:gm:modular-tail} uniformly in all finite volumes.
\end{proof}

\begin{lemma}[Small-resolvent tail]\label{lem:gm:small-quasilocal}
Under the same hypotheses, for $\ell>0$,
\begin{equation}\label{eq:gm:T-bound}
 \langle\rho^{1/2},\log(1+e^{-\ell-K})\rho^{1/2}\rangle
 \le C_\beta
 \exp\!\left[C_\beta(1+|\partial_{\mathrm e}A|)^{\theta/D}-c_\beta\ell^{\theta/D}\right].
\end{equation}

\end{lemma}

\begin{proof}
For real $u$,
\begin{equation}\label{eq:gm:scalar-log-tail}
 \log(1+e^{-\ell-u})
 \le e^{-\ell/2}
      +(\log2+|u|)\one_{\{u<-\ell/2\}}.
\end{equation}

On $u\ge-\ell/2$ this follows from $\log(1+w)\le w$; on the complementary set use $\log(1+e^v)\le\log2+v_+$.  For $x,b\ge0$, the scalar identity $x\one_{\{x>b\}}=b\one_{\{x>b\}}+\int_b^\infty\one_{\{x>s\}}\,ds$ gives, by spectral calculus,
\begin{equation}\label{eq:gm:spectral-first-tail-moment}
 \begin{aligned}
 &\langle\rho^{1/2},|K|\one_{\{|K|>b\}}\rho^{1/2}\rangle=b\langle\rho^{1/2},\one_{\{|K|>b\}}\rho^{1/2}\rangle
 +\int_b^\infty\langle\rho^{1/2},\one_{\{|K|>s\}}\rho^{1/2}\rangle\,ds.
 \end{aligned}
\end{equation}
Combining these inequalities with \eqref{eq:gm:modular-tail} at $b=\ell/2$ gives

\begin{equation}\label{eq:gm:stretched-small-resolvent-combination}
 \begin{aligned}
 \langle\rho^{1/2},\log(1+e^{-\ell-K})\rho^{1/2}\rangle \le e^{-\ell/2}
 +C_\beta e^{C_\beta(1+|\partial_{\mathrm e}A|)^{\theta/D}}
 \left[(1+\ell)e^{-c_\beta(\ell/2)^{\theta/D}}
 +\int_{\ell/2}^\infty e^{-c_\beta s^{\theta/D}}\,ds\right].
 \end{aligned}
\end{equation}
For the integral, use
\begin{equation}\label{eq:gm:stretched-tail-integral-bound}
 \begin{aligned}
 \int_b^\infty e^{-c_\beta s^{\theta/D}}\,ds
 &\le e^{-(c_\beta/2)b^{\theta/D}}
       \int_0^\infty e^{-(c_\beta/2)s^{\theta/D}}\,ds\\
 &\le C_\beta e^{-(c_\beta/2)b^{\theta/D}}.
 \end{aligned}
\end{equation}
The factor $1+\ell$ is absorbed by decreasing the decay constant.  The term $e^{-\ell/2}$ satisfies the same stretched-exponential bound because $\theta/D\le1$, with small $\ell$ covered by enlarging $C_\beta$.  This proves \eqref{eq:gm:T-bound}.
\end{proof}

The next two lemmas control the CMI error from local truncation and the spatial leakage of the truncated Gibbs pair. Both will be used for stretched-exponential and power-law interactions.

\subsection{A Gibbs-specific continuity estimate}

A trace-distance continuity bound for conditional entropy introduces a binary-entropy term and therefore a small-error logarithm \cite{Winter_2016}. 
For states connected by a bounded Hamiltonian perturbation at fixed temperature, the following direct estimate avoids that loss. 
\begin{lemma}[CMI along a Gibbs interpolation]\label{lem:gm:gibbs-cmi-lipschitz}
Let $H,W$ be Hermitian operators on $ABC$, put $d_A=\dim\mathcal H_A$, and
\begin{equation}\label{eq:gm:gibbs-interpolation-path}
 \rho_s=\frac{e^{-\beta(H+sW)}}{\Tr e^{-\beta(H+sW)}},\qquad 0\le s\le1.
\end{equation}
Then
\begin{equation}\label{eq:gm:gibbs-cmi-lipschitz}
 \bigl|I(A:C\mid B)_{\rho_1}-I(A:C\mid B)_{\rho_0}\bigr|
 \le 6\beta(1+\log d_A)\|W\|_\infty.
\end{equation}
\end{lemma}

\begin{proof}
Define, for a faithful tripartite state, the modular Markov defect
\begin{equation}\label{eq:gm:markov-log-defect}
 K_{\mathrm M}(\omega)
 =\log\omega-\log\omega_{AB}-\log\omega_{BC}+\log\omega_B.
\end{equation}
With $R_{A\mid Y}=\log\omega_{AY}-\log[(\one_A/d_A)\otimes\omega_Y]$, we have $K_{\mathrm M}(\omega)=R_{A\mid BC}-R_{A\mid B}$, where the second term is embedded on $ABC$.  Applying Lemma~\ref{lem:gm:conditional-logarithm} in Appendix~\ref{app:gm:conditional-logarithm} to $\omega$ and $\omega_{AB}$, followed by the triangle inequality in the state-weighted Hilbert--Schmidt norm, gives
\begin{equation}\label{eq:gm:markov-log-second-moment}
 \Tr\!\left[\omega K_{\mathrm M}(\omega)^2\right]
 \le4(e-1)(1+2\log d_A)^2.
\end{equation}

For Hermitian $X,Y$, write the Bogoliubov--Kubo--Mori inner product as
\begin{equation}\label{eq:gm:bkm-inner-product}
 \langle X,Y\rangle_{\mathrm{BKM},\rho}
 =\int_0^1\Tr(\rho^uX\rho^{1-u}Y)\,du.
\end{equation}
In an eigenbasis of $\rho$, the logarithmic mean is bounded by the arithmetic mean; hence
$\langle X,X\rangle_{\mathrm{BKM},\rho}\le\Tr(\rho X^2)$.
Differentiating the normalized Gibbs exponential by Duhamel's formula, \eqref{eq:gm:duhamel-normalized-gibbs}, gives
\begin{equation}\label{eq:gm:normalized-gibbs-derivative}
 \rho_s'=-\beta\int_0^1
 \rho_s^u\bigl(W-\Tr(\rho_sW)\one\bigr)\rho_s^{1-u}\,du.
\end{equation}
Since every marginal derivative has trace zero, differentiation of the four entropies gives
\begin{equation}\label{eq:gm:cmi-entropy-derivative}
 \frac{d}{ds}I(A:C\mid B)_{\rho_s}
 =\Tr\!\left[\rho_s'\bigl(\log\rho_s-\log\rho_{s,AB}
                  -\log\rho_{s,BC}+\log\rho_{s,B}\bigr)\right]
 =\Tr\!\left[\rho_s'K_{\mathrm M}(\rho_s)\right],
\end{equation}
where marginal logarithms are embedded by tensoring with the identity.
Substituting the Gibbs derivative yields
\begin{equation}\label{eq:gm:gibbs-cmi-derivative}
 \frac{d}{ds}I(A:C\mid B)_{\rho_s}
 =-\beta\left\langle W-\Tr(\rho_sW)\one,
                         K_{\mathrm M}(\rho_s)\right\rangle_{\mathrm{BKM},\rho_s}.
\end{equation}
Cauchy--Schwarz in the BKM inner product, the bound $\langle X,X\rangle_{\mathrm{BKM},\rho_s}\le\Tr(\rho_sX^2)$, \eqref{eq:gm:markov-log-second-moment}, and
$\Tr\rho_s(W-\Tr(\rho_sW)\one)^2\le\|W\|_\infty^2$ imply
\begin{equation}\label{eq:gm:gibbs-cmi-derivative-bound}
 \begin{aligned}
 \left|\frac{d}{ds}I(A:C\mid B)_{\rho_s}\right|
 &\le\beta\left[\Tr\rho_s\bigl(W-\Tr(\rho_sW)\one\bigr)^2\right]^{1/2}
 \left[\Tr\rho_s K_{\mathrm M}(\rho_s)^2\right]^{1/2}\\
 &\le2\sqrt{e-1}\,\beta(1+2\log d_A)\|W\|_\infty\\
 &\le6\beta(1+\log d_A)\|W\|_\infty.
 \end{aligned}
\end{equation}
Integration over $s\in[0,1]$ proves the lemma.
\end{proof}

\subsection{Modular leakage after local truncation}

We use the cut Gibbs pair and the general leakage bounds \eqref{eq:gm:duhamel-leakage} and \eqref{eq:gm:elementary-leakage} from Section~\ref{sec:gm:local-proof}.

\begin{lemma}[Weighted integrated modular leakage after local truncation]\label{lem:gm:powerlaw-collar-leakage}
Let $\Phi$ be an interaction on $\Z^D$ with
\begin{equation}\label{eq:gm:truncated-leakage-interaction-moments}
 J_0:=\sup_x\sum_{X\ni x}\|\Phi(X)\|_\infty<\infty,
 \qquad
 J_1:=\sup_x\sum_{X\ni x}|X|\|\Phi(X)\|_\infty<\infty.
\end{equation}
For $r\ge24$, choose an integer $1\le L\le r/24$ and define
\begin{equation}\label{eq:gm:local-truncation-definition}
 \Omega=\{x\in\Lambda:d(x,A)\le\lfloor r/3\rfloor\},
 \qquad
 H^{(L)}=H_\Lambda-
 \sum_{\substack{X\subseteq\Lambda\\X\cap\Omega\ne\varnothing\\\diam X>L}}\Phi(X).
\end{equation}
Let $V^{(L)}$ be the sum of retained interactions crossing $A\mid A^c$, and define
\begin{equation}\label{eq:gm:truncated-cut-gibbs-pair}
 H_0^{(L)}=H^{(L)}-V^{(L)},\qquad
 \rho^{(L)}=\frac{e^{-\beta H^{(L)}}}{\Tr e^{-\beta H^{(L)}}},\qquad
 \sigma^{(L)}=\frac{e^{-\beta H_0^{(L)}}}{\Tr e^{-\beta H_0^{(L)}}}.
\end{equation}
Let $\mathfrak q_{AB}^{(L)}$ denote the integrated modular leakage in \eqref{eq:gm:q-def} evaluated with $R=AB$, $\rho=\rho^{(L)}$, and $\sigma=\sigma^{(L)}$.  Suppose $w>0$ satisfies the weighted interaction bound
\begin{equation}\label{eq:gm:truncated-weighted-propagation-assumption}
 \sup_x\sum_y e^{w d(x,y)}
 \sum_{\substack{X\ni x,y\\X\subseteq\Omega\\\diam X\le L}}
 \|\Phi(X)\|_\infty\le J_*
\end{equation}
for a fixed constant $J_*>0$.  Then
\begin{equation}\label{eq:gm:powerlaw-collar-q}
 \mathfrak q_{AB}^{(L)}
 \le C_\beta|A|^2(1+r)^D e^{-c_\beta w r}.
\end{equation}
The constants depend only on $\beta,D,J_0,J_1,J_*$, and are independent of $w,L,\Lambda,A,B,C$.  In particular, $w$ may depend on $L$ provided the same $J_*$ bounds the weighted sum.  The choice $w=1/L$ always satisfies the hypothesis with $J_*=eJ_1$.
\end{lemma}

\begin{proof}
Remove the cut interaction from $H^{(L)}$ and write
\begin{equation}\label{eq:gm:truncated-hamiltonian-splitting}
 H^{(L)}=H_0^{(L)}+V^{(L)},
 \qquad
 H_0^{(L)}=H_{\rm in}+H_{\rm out}+W_\Omega.
\end{equation}
The last three sums contain, respectively, interactions inside $\Omega$, outside $\Omega$, and crossing its boundary.  Every retained interaction touching $\Omega$ has diameter at most $L$.
In particular, $V^{(L)}$ is supported inside the $L$-neighborhood of $A$ and
\begin{equation}\label{eq:gm:powerlaw-collar-strengths}
 \sum_{X\text{ in }V^{(L)}}|X|\|\Phi(X)\|_\infty\le |A|J_1,
 \qquad
 \sum_{Y\text{ in }W_\Omega}\|\Phi(Y)\|_\infty\le J_0|\Omega|.
\end{equation}
An individual support $X$ in $V^{(L)}$ and every support $Y$ in $W_\Omega$ satisfy
\begin{equation}\label{eq:gm:truncated-support-separation}
 d(X,Y)\ge\lfloor r/3\rfloor-2L
 \ge r/4-1\ge r/6,
\end{equation}
where we used $L\le r/24$ and $r\ge24$.
Indeed, a crossing support $Y$ has an exterior point at distance greater than $\lfloor r/3\rfloor$ from $A$, and all its points are within $L$ of that point; a support of $V^{(L)}$ is within $L$ of $A$.

The assumed weighted bound gives the interaction matrix estimate
\begin{equation}\label{eq:gm:weighted-propagation-matrix}
 a_{xy}=\sum_{\substack{X\ni x,y\\X\text{ in }H_{\rm in}}}\|\Phi(X)\|_\infty,
 \qquad
 \sup_x\sum_y e^{w d(x,y)}a_{xy}\le J_*.
\end{equation}
In the commutator-path expansion, the triangle inequality supplies the factor $e^{-w d(X,Y)}$, and each weighted step is bounded by $J_*$.  Summing the time-ordered series gives $e^{2J_*|t|}$.  Thus, for disjoint supports,
\begin{equation}\label{eq:gm:powerlaw-collar-LR}
 \|[\tau_t^{H_{\rm in}}(O_X),P_Y]\|_\infty
 \le 2\|O_X\|_\infty\|P_Y\|_\infty|X|
           e^{v|t|-w d(X,Y)},
 \qquad v=2J_*,
\end{equation}
where $\tau_t^G(O)=e^{itG}Oe^{-itG}$.
This is the standard Lieb--Robinson iteration \cite{NachtergaeleSimsYoung2019,ChenLucasYin2023}; no bound on $|Y|$ is needed because the row sum already includes all possible endpoints in $Y$.

The evolution under $H_{\rm in}+H_{\rm out}$ of $V^{(L)}$ is its evolution under $H_{\rm in}$ alone.  To bound the Duhamel integrand, expand both interactions and apply \eqref{eq:gm:powerlaw-collar-LR} and \eqref{eq:gm:truncated-support-separation}.  For $u\ge0$,
\begin{equation}\label{eq:gm:interface-commutator-sum}
 \begin{aligned}
 \|[W_\Omega,\tau_{\operatorname{sgn}(t)u}^{H_{\rm in}}(V^{(L)})]\|_\infty
 &\le\sum_{X\text{ in }V^{(L)}}\sum_{Y\text{ in }W_\Omega}
 \|[\Phi(Y),\tau_{\operatorname{sgn}(t)u}^{H_{\rm in}}(\Phi(X))]\|_\infty\\
 &\le2e^{vu-wr/6}
 \left(\sum_{X\text{ in }V^{(L)}}|X|\|\Phi(X)\|_\infty\right)
 \left(\sum_{Y\text{ in }W_\Omega}\|\Phi(Y)\|_\infty\right).
 \end{aligned}
\end{equation}
The two strength bounds in \eqref{eq:gm:powerlaw-collar-strengths} therefore give
\begin{equation}\label{eq:gm:interface-commutator-strength-bound}
 \|[W_\Omega,\tau_{\operatorname{sgn}(t)u}^{H_{\rm in}}(V^{(L)})]\|_\infty
 \le2|A|J_1J_0|\Omega|e^{vu-wr/6}.
\end{equation}
Duhamel's formula and unitarity of the exterior evolution give
\begin{align}
 \|\tau_t^{H_0^{(L)}}(V^{(L)})-\tau_t^{H_{\rm in}}(V^{(L)})\|_\infty
 &\le\int_0^{|t|}
      \|[W_\Omega,\tau_{\operatorname{sgn}(t)u}^{H_{\rm in}}(V^{(L)})]\|_\infty\,du
 \nonumber\\
 &\le C |A|J_1J_0|\Omega|\,|t|e^{v|t|-cwr}.
 \label{eq:gm:powerlaw-collar-dynamics}
\end{align}
Only commutators at the inner--outer interface are estimated.  No locality assumption is imposed on $H_{\rm out}$.

Up to a scalar phase, the relative modular evolution
$(\sigma^{(L)})^{it}(\rho^{(L)})^{-it}$ equals
$U(t):=e^{-i\beta tH_0^{(L)}}e^{i\beta tH^{(L)}}$, which satisfies
$U'(t)=i\beta\tau_{-\beta t}^{H_0^{(L)}}(V^{(L)})U(t)$.
Define the approximating unitary by
\begin{equation}\label{eq:gm:truncated-local-unitary}
 \widetilde U'(t)=i\beta\tau_{-\beta t}^{H_{\rm in}}(V^{(L)})\widetilde U(t),
 \qquad \widetilde U(0)=\one.
\end{equation}
Its generator is supported on $\Omega\subseteq AB$, so
$\mathsf E_{AB}(\widetilde U(t))=\widetilde U(t)$.
Duhamel's formula and \eqref{eq:gm:powerlaw-collar-dynamics} give
\begin{equation}\label{eq:gm:truncated-unitary-comparison}
 \begin{aligned}
 \|U(t)-\widetilde U(t)\|_\infty
 &\le\beta\int_0^{|t|}
 \bigl\|\tau_{-\beta\operatorname{sgn}(t)u}^{H_0^{(L)}}(V^{(L)})
 -\tau_{-\beta\operatorname{sgn}(t)u}^{H_{\rm in}}(V^{(L)})\bigr\|_\infty\,du\\
 &\le C\beta^2|A|J_1J_0|\Omega|e^{-cwr}
       \int_0^{|t|}u e^{v\beta u}\,du\\
 &\le C_\beta|A||\Omega|t^2e^{v\beta|t|-cwr}.
 \end{aligned}
\end{equation}
The last step uses $\int_0^{|t|}u e^{v\beta u}\,du\le(t^2/2)e^{v\beta|t|}$.
The scalar phase does not affect the leakage norm.  Contractivity of $\mathsf E_{AB}$ and $|\Omega|\le C|A|(1+r)^D$ now yield
\begin{equation}\label{eq:gm:truncated-unitary-leakage}
 \begin{aligned}
 \eta_{AB}^{(L)}(t)
 &=\|(\mathrm{id}-\mathsf E_{AB})(U(t)-\widetilde U(t))\|_\infty\\
 &\le2\|U(t)-\widetilde U(t)\|_\infty\\
 &\le C_\beta|A|^2(1+r)^D t^2e^{v\beta|t|-cwr}.
 \end{aligned}
\end{equation}
Also $\eta_{AB}^{(L)}(t)\le2$ by unitarity and contractivity, so
\begin{equation}\label{eq:gm:powerlaw-collar-eta}
 \eta_{AB}^{(L)}(t)
 \le\min\{2,
 C_\beta|A|^2(1+r)^D t^2 e^{v\beta|t|-cwr}\}.
\end{equation}
Split its integral against $1/(2|\sinh\pi t|)$ at
$T=cwr/(\pi+\beta v)$.
For $T\ge1$, on $|t|\le T$ we have $v\beta|t|-cwr\le-\pi T$, so the second bound and
$\int_\R t^2/|\sinh(\pi t)|\,dt<\infty$ give a contribution at most $C_\beta|A|^2(1+r)^D e^{-\pi T}$.
On $|t|>T$, the first bound and the exponential kernel tail give at most $C e^{-\pi T}$.
Hence
\begin{equation}\label{eq:gm:truncated-leakage-time-balance}
 \mathfrak q_{AB}^{(L)}
 \le C_\beta|A|^2(1+r)^D e^{-\pi T}
 =C_\beta|A|^2(1+r)^D
 \exp\!\left[-\frac{c\pi}{\pi+\beta v}wr\right].
\end{equation}
For $T<1$, the same bound follows after increasing $C_\beta$, since
$\mathfrak q_{AB}^{(L)}\le\beta g_{A\mid A^c}/2\le\beta J_0|A|/2$ by
\eqref{eq:gm:elementary-leakage}, while $e^{-\pi T}\ge e^{-\pi}$.
This proves \eqref{eq:gm:powerlaw-collar-q}, with constants uniform in $w$ and $L$.
\end{proof}

\subsection{Completion of the stretched-exponential proof}

\begin{proof}[Proof of Theorem~\ref{thm:gm:quasilocal}]
The constants $C_\beta,c_\beta>0$ may change from line to line.
Set
\begin{equation}\label{eq:gm:quasilocal-distance-threshold}
 r_0=\left\lceil24^{(D+1-\theta)/(D-\theta)}\right\rceil.
\end{equation}
This threshold depends only on $D,\theta$.  We first treat $r\ge r_0$ and handle $r<r_0$ at the end.

\emph{Step 1: Truncate and bound the CMI error.}
Choose an integer $1\le L\le r/24$, to be specified in Step~4, and set
\begin{equation}\label{eq:gm:quasilocal-truncated-hamiltonian}
 \Omega=\{x\in\Lambda:d(x,A)\le\lfloor r/3\rfloor\},
 \qquad
 H^{(L)}=H_\Lambda-
 \sum_{\substack{X\subseteq\Lambda\\X\cap\Omega\ne\varnothing\\\diam X>L}}\Phi(X).
\end{equation}
Then $\Omega\subseteq AB$ and $|\Omega|\le C_D|A|(1+r)^D$.
If $X\ni x$ has diameter greater than $L$, it contains a point $y$ with $d(x,y)>L/2$.  The aggregate decay bound \eqref{eq:gm:F-assumption} therefore gives
\begin{equation}\label{eq:gm:stretched-interaction-range-tail}
 \sup_x\sum_{\substack{X\ni x\\\diam X>L}}\|\Phi(X)\|_\infty
 \le J_{\mu,\theta}\sum_{|z|_1>L/2}e^{-\mu|z|_1^\theta}
 \le C e^{-cL^\theta}.
\end{equation}
Thus
\begin{equation}\label{eq:gm:quasilocal-truncation-norm}
 \|H_\Lambda-H^{(L)}\|_\infty\le C|A|(1+r)^D e^{-cL^\theta}.
\end{equation}
Writing $\rho^{(L)}$ for the Gibbs state of $H^{(L)}$, Lemma~\ref{lem:gm:gibbs-cmi-lipschitz} and $\log d_A\le C|A|$ imply
\begin{equation}\label{eq:gm:quasilocal-truncation-cmi}
 \bigl|I(A:C\mid B)_\rho-I(A:C\mid B)_{\rho^{(L)}}\bigr|
 \le C_\beta (1+|A|)^2
 (1+r)^D e^{-cL^\theta}.
\end{equation}

\emph{Step 2: Compare with the cut Gibbs state and separate the contributions.}
Let $\sigma^{(L)}$ be the cut Gibbs state in Lemma~\ref{lem:gm:powerlaw-collar-leakage}, obtained by removing the retained interactions crossing $A\mid BC$.  Lemma~\ref{lem:gm:cut} gives
\begin{equation}\label{eq:gm:truncated-cmi-comparison}
 I(A:C\mid B)_{\rho^{(L)}}
 \le\delta_{AB}(\rho^{(L)},\sigma^{(L)}).
\end{equation}
Let $K^{(L)}$ be the relative modular logarithm of this pair, as in \eqref{eq:gm:relative-modular}, and let $\mathfrak q_{AB}^{(L)}$ be its integrated modular leakage from \eqref{eq:gm:q-def}.  Applying \eqref{eq:gm:cutoff-inequality} and adding the truncation error yields, for every $\ell>0$,
\begin{equation}\label{eq:gm:quasilocal-four-contributions}
 \begin{aligned}
 I(A:C\mid B)_\rho\le{}&
 \underbrace{C_\beta(1+|A|)^2(1+r)^D e^{-cL^\theta}}_{\text{truncation error}}
 +\underbrace{\left\langle(\rho^{(L)})^{1/2},
 \log(1+e^{-\ell-K^{(L)}})(\rho^{(L)})^{1/2}\right\rangle}_{\text{small-resolvent contribution}}\\
 &+\underbrace{e^{-\ell}}_{\text{large-resolvent contribution}}
 +\underbrace{(\mathfrak q_{AB}^{(L)})^2(e^\ell+2\ell)}_{\text{leakage contribution}}.
 \end{aligned}
\end{equation}

\emph{Step 3: Estimate the leakage and the small-resolvent contribution.}
Choose $w_L=(\mu/2)L^{\theta-1}$.  For $0\le u\le L$, the inequality $L^{\theta-1}u\le u^\theta$ gives
$e^{w_Lu-\mu u^\theta}\le e^{-(\mu/2)u^\theta}$.  Hence
\begin{equation}\label{eq:gm:quasilocal-collar-weight}
 \sup_x\sum_y e^{w_Ld(x,y)}
 \sum_{\substack{X\ni x,y\\\diam X\le L}}\|\Phi(X)\|_\infty
 \le J_{\mu,\theta}\sum_{z\in\Z^D}e^{-(\mu/2)|z|_1^\theta}<\infty,
\end{equation}
with a bound independent of $L$.  The weighted support estimate \eqref{eq:gm:Ja} also bounds $J_0$ and $J_1$, so the hypotheses of Lemma~\ref{lem:gm:powerlaw-collar-leakage} hold uniformly in $L$.  Applying \eqref{eq:gm:powerlaw-collar-q} with $w=w_L=(\mu/2)L^{\theta-1}$ gives
\begin{equation}\label{eq:gm:quasilocal-collar-q}
 \mathfrak q_{AB}^{(L)}
 \le C_\beta|A|^2
 (1+r)^D\exp\!\left[-c_\beta\frac{r}{L^{1-\theta}}\right].
\end{equation}

Choose
\begin{equation}\label{eq:gm:quasilocal-resolvent-cutoff}
 \ell=c_0\frac{r}{L^{1-\theta}},
\end{equation}
where $c_0>0$ is sufficiently small relative to the decay constant in \eqref{eq:gm:quasilocal-collar-q}.  Squaring that bound and multiplying by $e^\ell+2\ell$ still leaves exponential decay in $r/L^{1-\theta}$.  Since $L\ge1$ implies $\ell\le c_0r$, the leakage and large-resolvent contributions satisfy
\begin{equation}\label{eq:gm:quasilocal-resolvent-leakage-error}
 e^{-\ell}+(\mathfrak q_{AB}^{(L)})^2(e^\ell+2\ell)
 \le C_\beta(1+|A|)^4(1+r)^{2D+1}
 e^{-c_\beta r/L^{1-\theta}}.
\end{equation}

Separately, deleting interaction terms preserves \eqref{eq:gm:F-assumption} and the weighted support and cut bounds.  Lemma~\ref{lem:gm:small-quasilocal} therefore applies uniformly in $L$ and gives
\begin{equation}\label{eq:gm:truncated-small-resolvent-bound}
 \begin{aligned}
 \left\langle(\rho^{(L)})^{1/2},
 \log(1+e^{-\ell-K^{(L)}})(\rho^{(L)})^{1/2}\right\rangle \le C_\beta\exp\!\left[
 C_\beta(1+|\partial_{\mathrm e}A|)^{\theta/D}
 -c_\beta\left(\frac{r}{L^{1-\theta}}\right)^{\theta/D}\right].
 \end{aligned}
\end{equation}
The boundary exponential enters through this spectral-tail estimate; the truncation and leakage factors remain polynomial in $|A|$.

\emph{Step 4: Balance the two remaining decay scales.}
Substitution into \eqref{eq:gm:quasilocal-four-contributions} gives
\begin{equation}\label{eq:gm:quasilocal-before-absorb}
 \begin{aligned}
 I(A:C\mid B)_\rho\le{}&
 C_\beta(1+|A|)^2(1+r)^D e^{-cL^\theta} +C_\beta\exp\!\left[
 C_\beta(1+|\partial_{\mathrm e}A|)^{\theta/D}
 -c_\beta\left(\frac{r}{L^{1-\theta}}\right)^{\theta/D}\right]\\
 &+C_\beta(1+|A|)^4(1+r)^{2D+1}
 e^{-c_\beta r/L^{1-\theta}}.
 \end{aligned}
\end{equation}
The competing scales are $L^\theta$ from truncation and $(r/L^{1-\theta})^{\theta/D}$ from the small-resolvent tail.
Choose $L=\lfloor r^{1/(D+1-\theta)}\rfloor$.  Since $D\ge2$ and $r\ge r_0$, we have
\begin{equation}\label{eq:gm:quasilocal-range-admissibility}
 \tfrac12 r^{1/(D+1-\theta)}\le L
 \le r^{1/(D+1-\theta)}\le r/24.
\end{equation}
Thus $1\le L\le r/24$, as required for the truncation estimate.  This is where the threshold $r_0$ is used; it also ensures $r\ge24$ as required by Lemma~\ref{lem:gm:powerlaw-collar-leakage}.
The two decay scales agree:
\begin{equation}\label{eq:gm:quasilocal-scale-balance}
 L^\theta\asymp
 \left(\frac{r}{L^{1-\theta}}\right)^{\theta/D}
 \asymp r^{\theta/(D+1-\theta)}.
\end{equation}
The leakage and large-resolvent contributions decay faster, since $r/L^{1-\theta}\asymp r^{D/(D+1-\theta)}$.
It remains to put the prefactors into the form of the theorem.  Since $D\ge2$, lattice isoperimetry and the on-site dimension bound give
\begin{equation}\label{eq:gm:lattice-isoperimetry-dimension-bound}
 |A|\le C_D(1+|\partial_{\mathrm e}A|)^{D/(D-1)},
 \qquad \log d_A\le C|A|.
\end{equation}
Consequently, for every fixed $m\ge0$,
\begin{equation}\label{eq:gm:quasilocal-region-absorption}
 (1+|A|)^m(1+\log d_A)
 \le C_m\exp\!\left[C_m(1+|\partial_{\mathrm e}A|)^{\theta/D}\right].
\end{equation}
Apply this estimate to the polynomial region factors in \eqref{eq:gm:quasilocal-before-absorb} and absorb the powers of $1+r$ into the decaying exponentials.  Since $1+r\le2r$ here, adjusting constants proves \eqref{eq:gm:quasilocal-main} for every $r\ge r_0$.
For $r<r_0$, the elementary bound $I(A:C\mid B)_\rho\le2\log d_A\le C|A|$ and \eqref{eq:gm:quasilocal-region-absorption} give
\begin{equation}\label{eq:gm:quasilocal-bounded-distance}
 \begin{aligned}
 I(A:C\mid B)_\rho
 &\le C\exp\!\left[C(1+|\partial_{\mathrm e}A|)^{\theta/D}\right]\\
 &\le C e^{c_\beta(1+r_0)^{\theta/(D+1-\theta)}}
 \exp\!\left[C(1+|\partial_{\mathrm e}A|)^{\theta/D}
 -c_\beta(1+r)^{\theta/(D+1-\theta)}\right].
 \end{aligned}
\end{equation}

The additional prefactor is independent of $\Lambda,A,B,C$.  Enlarging $C_\beta$ to dominate both $C$ and $C e^{c_\beta(1+r_0)^{\theta/(D+1-\theta)}}$ proves \eqref{eq:gm:quasilocal-main} also for $r<r_0$, with the same $c_\beta$.
\end{proof}

\begin{remark}
Truncating the Hamiltonian before applying the Lieb--Robinson bound improves the distance-decay exponent obtained by our argument.
A standard Lieb--Robinson estimate applied directly to the original stretched-exponential interaction
\cite[Theorem~3.1 and Appendix]{NachtergaeleSimsYoung2019}
bounds $\mathfrak q_{AB}$ by $e^{-c_\beta r^\theta}$ up to polynomial prefactors.
Choosing the resolvent cutoff $\ell\asymp r^\theta$ and using the same small-resolvent estimate then gives a CMI bound with distance decay $e^{-c_\beta r^{\theta^2/D}}$.
Truncation instead permits $\ell\asymp r/L^{1-\theta}$; balancing the spectral contribution against the truncation error gives the exponent $\theta/(D+1-\theta)$.
This exponent is strictly larger than $\theta^2/D$ for $0<\theta<1$, while the two coincide at $\theta=1$.
\end{remark}

\section{Power-law interactions: spectral tails and spatial truncation}\label{sec:gm:powerlaw-proof}

This section proves Theorem~\ref{thm:gm:powerlaw}.  We use the same cut Gibbs pair, relative modular logarithm $K$, and integrated modular leakage $\mathfrak q_{AB}$ as in Section~\ref{sec:gm:modular}.  
We reuse the Gibbs-interpolation estimate of Lemma~\ref{lem:gm:gibbs-cmi-lipschitz} and combine bounds on repeated modular actions with support-size tails to estimate the spectral distribution of $K$. 
For $D<\zeta\le2D$, the direct spatial leakage bound permits a resolvent cutoff proportional to $\log(1+r)$; combining this with the spectral-tail estimate gives inverse-logarithmic CMI decay as shown in \eqref{eq:gm:powerlaw-log-main}--\eqref{eq:gm:powerlaw-log-endpoint},
with the squared double-logarithm correction at $\zeta=2D$.
For $\zeta>2D$, we truncate long interactions near $A$ and choose the truncation range to balance the CMI error with the spectral-tail contribution. This gives the algebraic decay of CMI. The threshold $\zeta>2D$ comes from the truncation error $C_\beta(1+|A|)^2r^D L^{D-\zeta}$ in \eqref{eq:gm:powerlaw-truncation-cmi}: choosing $L\asymp r^\alpha$ with $0<\alpha<1$ makes this bound decay only if $\alpha(\zeta-D)>D$, which is possible precisely when $\zeta>2D$.

\subsection{Diameter tails and endpoint support-size tails}

\begin{lemma}[Interaction tails]\label{lem:gm:powerlaw-support-tail}
Under \eqref{eq:gm:powerlaw-assumption}, suppose $\zeta>D$.  For $L,N\ge1$,
\begin{align}
 \sup_x\sum_{\substack{X\ni x\\\diam X>L}}\|\Phi(X)\|_\infty
 &\le C J_\zeta(1+L)^{D-\zeta},\label{eq:gm:powerlaw-diameter-tail}\\
 \sup_x\sum_{\substack{X\ni x\\|X|>N}}\|\Phi(X)\|_\infty
 &\le C J_\zeta(1+N)^{-\zeta/D}.\label{eq:gm:powerlaw-cardinality-tail}
\end{align}
In particular,
\begin{equation}\label{eq:gm:powerlaw-J01}
 J_0:=\sup_x\sum_{X\ni x}\|\Phi(X)\|_\infty\le J_\zeta,
 \qquad
 J_1:=\sup_x\sum_{X\ni x}|X|\|\Phi(X)\|_\infty<\infty.
\end{equation}
For the natural decomposition of the cut interaction,
\begin{equation}\label{eq:gm:powerlaw-cut-tail}
 g_{A\mid A^c}\le J_0|A|,
 \qquad
 \sum_{\substack{X\text{ crossing }A\mid A^c\\|X|>N}}
        \|\Phi(X)\|_\infty
 \le C J_\zeta|A|(1+N)^{-\zeta/D}.
\end{equation}
All constants are uniform under deletion of interaction terms.
\end{lemma}

\begin{proof}
If $X\ni x$ has diameter greater than $L$, some $y\in X$ has $d(x,y)>L/2$.
Summing \eqref{eq:gm:powerlaw-assumption} over these $y$ proves
\eqref{eq:gm:powerlaw-diameter-tail}, since the lattice tail is $O((1+L)^{D-\zeta})$.

For \eqref{eq:gm:powerlaw-cardinality-tail}, choose $R\asymp N^{1/D}$ so that the lattice ball $B_R(x)$ contains at most $N/2$ sites, treating bounded $N$ by increasing $C$.
Every $X\ni x$ with $|X|>N$ contains at least $N/2$ points outside $B_R(x)$.
Consequently,
\begin{equation}\label{eq:gm:powerlaw-cardinality-tail-derivation}
 \sum_{\substack{X\ni x\\|X|>N}}\|\Phi(X)\|_\infty
 \le\frac2N\sum_{y:d(x,y)>R}\sum_{X\ni x,y}\|\Phi(X)\|_\infty
 \le\frac{CJ_\zeta}{N}R^{D-\zeta}
 \le CJ_\zeta N^{-\zeta/D}.
\end{equation}
The factor $N^{-1}$ uses all the distant sites in a large support, rather than selecting just one.  The diagonal case $x=y$ gives $J_0\le J_\zeta$.  Applying \eqref{eq:gm:powerlaw-assumption} to each pair $x,y$ gives
\begin{equation}\label{eq:gm:powerlaw-first-support-moment}
 \begin{aligned}
 J_1
 &=\sup_x\sum_y\sum_{X\ni x,y}\|\Phi(X)\|_\infty
 \le J_\zeta\sup_x\sum_{y\in\Z^D}(1+|y-x|_1)^{-\zeta}
 =J_\zeta\sum_{z\in\Z^D}(1+|z|_1)^{-\zeta}.
 \end{aligned}
\end{equation}
Since the number of sites with $|z|_1=n$ is at most $C_D(1+n)^{D-1}$,
\begin{equation}\label{eq:gm:powerlaw-lattice-summability}
 \sum_{z\in\Z^D}(1+|z|_1)^{-\zeta}
 \le C_D\sum_{n=0}^\infty(1+n)^{D-1-\zeta}<\infty,
\end{equation}
where convergence follows from $\zeta>D$.  Thus $J_1<\infty$.
Every support $X$ crossing $A\mid A^c$ contains at least one site of $A$, so its norm is counted at least once in the sum over $x\in A$.  Hence
\begin{equation}\label{eq:gm:cut-strength-volume-bound}
 g_{A\mid A^c}
 =\sum_{X\text{ crossing }A\mid A^c}\|\Phi(X)\|_\infty
 \le\sum_{x\in A}\sum_{X\ni x}\|\Phi(X)\|_\infty
 \le J_0|A|.
\end{equation}
Restricting to supports with $|X|>N$ and applying \eqref{eq:gm:powerlaw-cardinality-tail} at each site gives
\begin{equation}\label{eq:gm:powerlaw-cut-tail-counting}
 \sum_{\substack{X\text{ crossing }A\mid A^c\\|X|>N}}\|\Phi(X)\|_\infty
 \le\sum_{x\in A}\sum_{\substack{X\ni x\\|X|>N}}\|\Phi(X)\|_\infty
 \le C J_\zeta|A|(1+N)^{-\zeta/D}.
\end{equation}
These are the two bounds in \eqref{eq:gm:powerlaw-cut-tail}.
Every argument bounds a sum of nonnegative interaction norms, proving the deletion-uniformity assertion.
\end{proof}

\subsection{Support-size tails under repeated modular action}

In this subsection, a local decomposition $O=\sum_XO_X$ is kept explicitly.  Repeated terms on the same support may be kept separately or combined using the triangle inequality.  Define its coefficient tail by
\begin{equation}\label{eq:gm:coefficient-support-tail-definition}
 T_O(N)=\sum_{|X|>N}\|O_X\|_\infty.
\end{equation}
\begin{lemma}[One-power loss under a commutator]\label{lem:gm:powerlaw-tail-loss}
Suppose $M_O:=\sum_X\|O_X\|_\infty<\infty$ and
$T_O(N)\le C_O(1+N)^{-s}$ with $s>1$.
Then the natural decomposition of $[O,H_0]$ has finite total coefficient sum and
\begin{equation}\label{eq:gm:powerlaw-commutator-tail}
 T_{[O,H_0]}(N)
 \le C(M_O+C_O)\left[(1+N)^{-(s-1)}+(1+N)^{-\zeta/D}\right].
\end{equation}
The product $OV$ has a decomposition with
\begin{equation}\label{eq:gm:powerlaw-product-tail}
 T_{OV}(N)\le C|A|\left[C_O(1+N)^{-s}+M_O(1+N)^{-\zeta/D}\right].
\end{equation}
Here $C$ depends only on $D,J_\zeta,\zeta,s$.  The total coefficient sums satisfy
\begin{equation}\label{eq:gm:powerlaw-total-coefficients}
 \sum_Z\|[O,H_0]_Z\|_\infty\le C(M_O+C_O),
 \qquad
 \sum_Z\|(OV)_Z\|_\infty\le J_\zeta|A|M_O.
\end{equation}
\end{lemma}

\begin{proof}
The output support of a commutator between $O_X$ and $\Phi(Y)$ is contained in $X\cup Y$.  If its size exceeds $N$, then $|X|>N/2$ or $|Y|>N/2$.
Only pairs with $X\cap Y\ne\varnothing$ contribute, and
$\|[O_X,\Phi(Y)]\|_\infty\le2\|O_X\|_\infty\|\Phi(Y)\|_\infty$.
Splitting according to which support is large gives
\begin{equation}\label{eq:gm:powerlaw-tail-split}
 \begin{aligned}
 T_{[O,H_0]}(N)
 \le{}&2\sum_{|X|>N/2}\|O_X\|_\infty
          \sum_{Y:Y\cap X\ne\varnothing}\|\Phi(Y)\|_\infty
 +2\sum_X\|O_X\|_\infty
          \sum_{\substack{Y:Y\cap X\ne\varnothing\\|Y|>N/2}}\|\Phi(Y)\|_\infty\\
 \le{}&2\sum_{|X|>N/2}\|O_X\|_\infty
          \sum_{x\in X}\sum_{Y\ni x}\|\Phi(Y)\|_\infty
 +2\sum_X\|O_X\|_\infty
          \sum_{x\in X}\sum_{\substack{Y\ni x\\|Y|>N/2}}\|\Phi(Y)\|_\infty\\
 \le{}&2J_0\sum_{|X|>N/2}|X|\|O_X\|_\infty
 +2\left(\sup_x\sum_{\substack{Y\ni x\\|Y|>N/2}}\|\Phi(Y)\|_\infty\right)
          \sum_X|X|\|O_X\|_\infty.
 \end{aligned}
\end{equation}
The sum over $x\in X$ counts every interaction meeting $X$ at least once.  Summing over all interaction terms also bounds the sum over those retained in $H_0$.
The identity
\begin{equation}\label{eq:gm:coefficient-first-tail-moment}
 \begin{aligned}
 \sum_{|X|>N}|X|\|O_X\|_\infty
 &=\sum_{|X|>N}\left(N+\int_N^{|X|}du\right)\|O_X\|_\infty\\
 &=N T_O(N)+\int_N^\infty T_O(u)\,du
 \end{aligned}
\end{equation}
shows that the first term in \eqref{eq:gm:powerlaw-tail-split} is $O((1+N)^{-(s-1)})$ and that the full first support moment is finite.  Lemma~\ref{lem:gm:powerlaw-support-tail} bounds the second term by $O((1+N)^{-\zeta/D})$.

For $OV$, write $V=\sum_YV_Y$ for the natural cut decomposition, with $V_Y=\Phi(Y)$ when $Y$ crosses $A\mid A^c$ and $V_Y=0$ otherwise.  Each product $O_XV_Y$ is supported on $X\cup Y$ and obeys
$\|O_XV_Y\|_\infty\le\|O_X\|_\infty\|V_Y\|_\infty$, even when $X$ and $Y$ are disjoint.
Again, $|X\cup Y|>N$ implies $|X|>N/2$ or $|Y|>N/2$, so
\begin{equation}\label{eq:gm:powerlaw-product-tail-derivation}
 \begin{aligned}
 T_{OV}(N)
 &\le\sum_{\substack{X,Y\\|X\cup Y|>N}}\|O_X\|_\infty\|V_Y\|_\infty\\
 &\le\left(\sum_{|X|>N/2}\|O_X\|_\infty\right)\left(\sum_Y\|V_Y\|_\infty\right)
 +\left(\sum_X\|O_X\|_\infty\right)\left(\sum_{|Y|>N/2}\|V_Y\|_\infty\right)\\
 &=g_{A\mid A^c}T_O(N/2)+M_OT_V(N/2)\\
 &\le C|A|\left[C_O(1+N)^{-s}+M_O(1+N)^{-\zeta/D}\right].
 \end{aligned}
\end{equation}
The last inequality uses the assumed bound on $T_O$, together with $g_{A\mid A^c}\le J_\zeta|A|$ and the bound on $T_V$ in \eqref{eq:gm:powerlaw-cut-tail}.  
Finally, for \eqref{eq:gm:powerlaw-total-coefficients}, summing over all output supports $Z$, with terms grouped by $Z=X\cup Y$, gives
\begin{equation}\label{eq:gm:commutator-total-coefficient-bound}
 \begin{aligned}
 \sum_Z\|[O,H_0]_Z\|_\infty
 &\le2\sum_X\|O_X\|_\infty\sum_{Y:Y\cap X\ne\varnothing}\|\Phi(Y)\|_\infty\\
 &\le2J_0\sum_X|X|\|O_X\|_\infty\\
 &=2J_0\int_0^\infty T_O(u)\,du \le C(M_O+C_O).
 \end{aligned}
\end{equation}
Here we used $T_O(u)\le M_O$ for $0\le u<1$ and the assumed tail bound for $u\ge1$.
For the product decomposition, the corresponding estimate is
\begin{equation}\label{eq:gm:product-total-coefficient-bound}
 \begin{aligned}
 \sum_Z\|(OV)_Z\|_\infty
 &\le\sum_{X,Y}\|O_XV_Y\|_\infty\\
 &\le\left(\sum_X\|O_X\|_\infty\right)\left(\sum_Y\|V_Y\|_\infty\right)\\
 &=M_Og_{A\mid A^c}
 \le J_\zeta|A|M_O.
 \end{aligned}
\end{equation}
\end{proof}

Recall the map $\mathcal D$ and the operators $O_j=\mathcal D^j(\one)$ defined in \eqref{eq:gm:powerlaw-D-map}--\eqref{eq:gm:modular-iterates-definition}. We write $O_j=\sum_XO_{j,X}$ for the local decompositions constructed below, so $K^j\rho^{1/2}=O_j\rho^{1/2}$. 
\begin{lemma}[Iterated modular-action bounds]\label{lem:gm:powerlaw-iterates}
Under \eqref{eq:gm:powerlaw-assumption}, fix $\beta>0$ and let $\mathcal D$ be defined by \eqref{eq:gm:powerlaw-D-map}.  The iterates $O_j$ admit local decompositions $O_j=\sum_XO_{j,X}$ satisfying, for $N\ge1$,
\begin{equation}\label{eq:gm:powerlaw-iterate-tail}
 \sum_{|X|>N}\|O_{j,X}\|_\infty
 \le C_{\beta,j}|A|^j(1+N)^{-(\zeta/D-j+1)},
 \qquad 1\le j\le\lceil \zeta/D\rceil,
\end{equation}
and
\begin{equation}\label{eq:gm:powerlaw-iterate-total-coefficients}
 \sum_X\|O_{j,X}\|_\infty\le C_{\beta,j}|A|^j.
\end{equation}
Here $C_{\beta,j}$ depends only on $\beta,D,J_\zeta,\zeta,j$, independently of the regions and finite volume.
\end{lemma}

\begin{proof}
We prove both bounds simultaneously by induction.
For $j=1$, we have $O_1=\beta V+(\log Z-\log Z_0)\one$.  Assigning the scalar term to the empty support, \eqref{eq:gm:powerlaw-cut-tail} and $|\log Z-\log Z_0|\le\beta g_{A\mid A^c}$ give
\begin{equation}\label{eq:gm:powerlaw-iterate-base-case}
 \begin{aligned}
 \sum_{|X|>N}\|O_{1,X}\|_\infty
 &=\beta T_V(N) \le C_\beta|A|(1+N)^{-\zeta/D},\\
 \sum_X\|O_{1,X}\|_\infty
 &\le\beta g_{A\mid A^c}+|\log Z-\log Z_0|\le2\beta J_\zeta|A|.
 \end{aligned}
\end{equation}
Now assume both bounds hold at order $j$, where $1\le j<\lceil\zeta/D\rceil$.
The input tail exponent is $\zeta/D-j+1>1$, so Lemma~\ref{lem:gm:powerlaw-tail-loss} applies to $O=O_{j}$.
Since $\zeta/D-j+1\le\zeta/D$, its commutator and product estimates, together with the bound on the coefficient of $(\log Z-\log Z_0)O_j$, yield
\begin{equation}\label{eq:gm:powerlaw-iterate-induction-inputs}
 \begin{aligned}
 \beta T_{[O_{j},H_0]}(N)
 &\le C_{\beta,j}|A|^j(1+N)^{-(\zeta/D-j)},\\
 \beta T_{O_{j}V}(N)
 &\le C_{\beta,j}|A|^{j+1}(1+N)^{-(\zeta/D-j+1)},\\
 |\log Z-\log Z_0|\,T_{O_{j}}(N)
 &\le C_{\beta,j}|A|^{j+1}(1+N)^{-(\zeta/D-j+1)}.
 \end{aligned}
\end{equation}
Thus only the commutator lowers the tail exponent by one; the terms $O_jV$ and $(\log Z-\log Z_0)O_j$ preserve the exponent and contribute one further factor of $|A|$.
Adding the three contributions in \eqref{eq:gm:powerlaw-D-map} gives
\begin{equation}\label{eq:gm:powerlaw-iterate-tail-induction}
 \begin{aligned}
 \sum_{|X|>N}\|O_{j+1,X}\|_\infty
 &\le C_{\beta,j}\left[|A|^j(1+N)^{-(\zeta/D-j)}
       +|A|^{j+1}(1+N)^{-(\zeta/D-j+1)}\right]\\
 &\le C_{\beta,j+1}|A|^{j+1}(1+N)^{-(\zeta/D-j)},
 \end{aligned}
\end{equation}
where we used $|A|\ge1$ and $1+N\ge1$.
Similarly, \eqref{eq:gm:powerlaw-total-coefficients} controls the total coefficient sums of the commutator and product, while the term $(\log Z-\log Z_0)O_j$ multiplies the input coefficient sum by at most $\beta J_\zeta|A|$.  Hence
\begin{equation}\label{eq:gm:powerlaw-iterate-coefficient-induction}
 \begin{aligned}
 \sum_X\|O_{j+1,X}\|_\infty
 &\le C_{\beta,j}\left(|A|^j+|A|^{j+1}\right)\le C_{\beta,j+1}|A|^{j+1}.
 \end{aligned}
\end{equation}
This completes the induction.  At $j=\lceil\zeta/D\rceil$, the output tail exponent is $\zeta/D-\lceil\zeta/D\rceil+1\in(0,1]$; no further commutator is estimated by Lemma~\ref{lem:gm:powerlaw-tail-loss}.
\end{proof}

\begin{lemma}[Endpoint modular spectral tail]\label{lem:gm:powerlaw-modular-tail}
For $R\ge1, K=\log\Delta$,
\begin{equation}\label{eq:gm:powerlaw-spectral-tail}
 \langle\rho^{1/2},\one_{\{|K|\ge R\}}\rho^{1/2}\rangle
 \le C_\beta|A|^{b_\zeta+1}R^{-2(1+\zeta/D)}
 \begin{cases}
  1,&\zeta/D\notin\mathbb N,\\
  [\log(e+R)]^2,&\zeta/D\in\mathbb N.
 \end{cases}
\end{equation}
Consequently, for $\ell\ge1$,
\begin{equation}\label{eq:gm:powerlaw-small-tail}
 \left\langle\rho^{1/2},\log(1+e^{-\ell-K})\rho^{1/2}\right\rangle
 \le C_\beta|A|^{b_\zeta+1}\ell^{-b_\zeta}
 \begin{cases}
  1,&\zeta/D\notin\mathbb N,\\
  [\log(e+\ell)]^2,&\zeta/D\in\mathbb N,
 \end{cases}
\end{equation}
where $b_\zeta=1+2\zeta/D$.  The estimates remain valid, with unchanged constants, after arbitrary interaction deletions.
\end{lemma}

\begin{proof}
Put $m=\lceil\zeta/D\rceil$ and $s=\zeta/D-m+1\in(0,1]$.
For each term $O_{m,X}$ in the decomposition supplied by
Lemma~\ref{lem:gm:powerlaw-iterates},
\eqref{eq:gm:K-action} and the per-site interaction bound imply
\begin{equation}\label{eq:gm:modular-action-single-support-bound}
 \|K(O_{m,X}\rho^{1/2})\|_2
 \le\bigl(2\beta J_0|X|+2\beta g_{A\mid A^c}\bigr)\|O_{m,X}\|_\infty.
\end{equation}
For $R\ge1$, spectral calculus gives
\begin{equation}\label{eq:gm:direct-spectral-projection}
 \|\one_{\{|K|\ge R\}}v\|_2\le R^{-1}\|Kv\|_2,
 \qquad \|\one_{\{|K|\ge R\}}v\|_2\le\|v\|_2.
\end{equation}
Split the decomposition of $K^m\rho^{1/2}$ into supports of size at most $R$ and larger than $R$. Apply the first bound to the former and the second to the latter. Using \eqref{eq:gm:modular-action-single-support-bound}, $g_{A\mid A^c}\le J_\zeta|A|$, and the triangle inequality gives
\begin{equation}\label{eq:gm:direct-spectral-support-split}
 \begin{aligned}
 \|\one_{\{|K|\ge R\}}K^m\rho^{1/2}\|_2
 &\le\frac{C_\beta}{R}\sum_{|X|\le R}(|A|+|X|)
 \|O_{m,X}\|_\infty
 +\sum_{|X|>R}\|O_{m,X}\|_\infty.
 \end{aligned}
\end{equation}
Lemma~\ref{lem:gm:powerlaw-iterates} bounds the truncated first support moment by
\begin{equation}\label{eq:gm:truncated-support-first-moment}
 \begin{aligned}
 \sum_{|X|\le R}|X|\|O_{m,X}\|_\infty
 &\le\int_0^R\sum_{|X|>u}\|O_{m,X}\|_\infty\,du\\
 &\le C_\beta|A|^m\int_0^R(1+u)^{-s}\,du\\
 &\le C_\beta|A|^m
 \begin{cases}
 R^{1-s},&0<s<1,\\
 \log(e+R),&s=1.
 \end{cases}
 \end{aligned}
\end{equation}
Here the first inequality follows by integrating the indicator $\one_{\{u<|X|\}}$; the total coefficient bound covers $0\le u<1$. Combining this with the total coefficient and support-tail bounds yields
\begin{equation}\label{eq:gm:direct-spectral-tail-preliminary}
 \|\one_{\{|K|\ge R\}}K^m\rho^{1/2}\|_2
 \le C_\beta\left[\frac{|A|^{m+1}}R+|A|^m
 \begin{cases}
 R^{-s},&0<s<1,\\
 R^{-1}\log(e+R),&s=1
 \end{cases}\right].
\end{equation}
For $R\ge|A|$, since $0<s\le1$,
\begin{equation}\label{eq:gm:direct-spectral-region-absorption}
 \frac{|A|^{m+1}}R
 \le |A|^m\left(\frac{|A|}{R}\right)^s
 =|A|^{m+s}R^{-s}.
\end{equation}
For $1\le R<|A|$, the total coefficient bound instead gives
\begin{equation}\label{eq:gm:direct-spectral-small-threshold}
 \|\one_{\{|K|\ge R\}}K^m\rho^{1/2}\|_2
 \le C_\beta|A|^m\le C_\beta|A|^{m+s}R^{-s}.
\end{equation}
Thus, using $|A|\ge1$ and $\log(e+R)\ge1$, squaring the bounds in both ranges and using $2(m+s)=b_\zeta+1$ proves
\begin{equation}\label{eq:gm:powerlaw-weighted-spectral-tail}
 \begin{aligned}
 \langle\rho^{1/2},\one_{\{|K|\ge R\}}|K|^{2m}\rho^{1/2}\rangle
 &=\|\one_{\{|K|\ge R\}}K^m\rho^{1/2}\|_2^2\\
 &\le C_\beta|A|^{b_\zeta+1}
 \begin{cases}
 R^{-2s},&0<s<1,\\
 R^{-2}[\log(e+R)]^2,&s=1.
 \end{cases}
 \end{aligned}
\end{equation}
Since $\one_{\{|K|\ge R\}}|K|^{2m}\ge R^{2m}\one_{\{|K|\ge R\}}$, division by $R^{2m}$ proves \eqref{eq:gm:powerlaw-spectral-tail} for every $R\ge1$, since $m+s=\zeta/D+1$.

For $R\ge1$, the first tail-moment identity \eqref{eq:gm:spectral-first-tail-moment} with $b=R$ reads
\begin{equation}\label{eq:gm:powerlaw-first-spectral-tail-moment}
 \langle\rho^{1/2},|K|\one_{\{|K|>R\}}\rho^{1/2}\rangle
 =R\langle\rho^{1/2},\one_{\{|K|>R\}}\rho^{1/2}\rangle
 +\int_R^\infty\langle\rho^{1/2},\one_{\{|K|>v\}}\rho^{1/2}\rangle\,dv.
\end{equation}
Since $2(1+\zeta/D)=b_\zeta+1$, the required scalar integrals are
\begin{equation}\label{eq:gm:powerlaw-tail-integrals}
 \begin{aligned}
 \int_R^\infty v^{-(b_\zeta+1)}\,dv
 &=\frac{R^{-b_\zeta}}{b_\zeta},\\
 \int_R^\infty v^{-(b_\zeta+1)}[\log(e+v)]^2\,dv
 &\le C R^{-b_\zeta}[\log(e+R)]^2.
 \end{aligned}
\end{equation}
The second estimate follows by substituting $v=Rx$ and using $\log(e+Rx)\le\log(e+R)+\log x$ for $x\ge1$.
Substituting \eqref{eq:gm:powerlaw-spectral-tail} into the tail-moment identity therefore gives
\begin{equation}\label{eq:gm:powerlaw-first-spectral-moment-bound}
 \langle\rho^{1/2},|K|\one_{\{|K|>R\}}\rho^{1/2}\rangle
 \le C_\beta|A|^{b_\zeta+1}R^{-b_\zeta}
 \begin{cases}
 1,&\zeta/D\notin\mathbb N,\\
 [\log(e+R)]^2,&\zeta/D\in\mathbb N.
 \end{cases}
\end{equation}
For $\ell\ge2$, apply \eqref{eq:gm:scalar-log-tail} by spectral calculus and enlarge the negative spectral tail to $\{|K|>\ell/2\}$.  Using \eqref{eq:gm:powerlaw-first-spectral-moment-bound} with $R=\ell/2$ and \eqref{eq:gm:powerlaw-spectral-tail} for the term proportional to $\log2$ gives
\begin{equation}\label{eq:gm:powerlaw-small-resolvent-combination}
 \begin{aligned}
 &\left\langle\rho^{1/2},\log(1+e^{-\ell-K})\rho^{1/2}\right\rangle \le e^{-\ell/2}+C_\beta|A|^{b_\zeta+1}
 \left[\ell^{-(b_\zeta+1)}+\ell^{-b_\zeta}\right]
 \begin{cases}
 1,&\zeta/D\notin\mathbb N,\\
 [\log(e+\ell)]^2,&\zeta/D\in\mathbb N.
 \end{cases}
 \end{aligned}
\end{equation}
The bounds $e^{-\ell/2}\le C \ell^{-b_\zeta}$, $\ell^{-(b_\zeta+1)}\le \ell^{-b_\zeta}$, and $|A|\ge1$ prove \eqref{eq:gm:powerlaw-small-tail} for $\ell\ge2$.
For $1\le \ell<2$, \eqref{eq:gm:powerlaw-spectral-tail} gives
\begin{equation}\label{eq:gm:powerlaw-bounded-cutoff-estimate}
 \begin{aligned}
 \left\langle\rho^{1/2},\log(1+e^{-\ell-K})\rho^{1/2}\right\rangle
 &\le\log2+\langle\rho^{1/2},|K|\rho^{1/2}\rangle\\
 &\le\log2+1+\int_1^\infty\langle\rho^{1/2},\one_{\{|K|>v\}}\rho^{1/2}\rangle\,dv\\
 &\le C_\beta|A|^{b_\zeta+1}.
 \end{aligned}
\end{equation}
Enlarging $C_\beta$ extends \eqref{eq:gm:powerlaw-small-tail} to this bounded range and completes the proof.
\end{proof}

\subsection{Spatial leakage for summable power-law profiles}\label{sec:gm:spatial}

\begin{lemma}[Power-law modular leakage]\label{lem:gm:F-leakage}
Under \eqref{eq:gm:powerlaw-assumption},
\begin{equation}\label{eq:gm:F-eta}
 \eta_{AB}(t)
 \le\min\left\{2,\ C\beta J_\zeta|t|e^{v\beta|t|}
        \sum_{a\in A,c\in C}(1+d(a,c))^{-\zeta}\right\},
 \qquad v=C_{D,\zeta}J_\zeta.
\end{equation}
Consequently, for $r=d(A,C)$,
\begin{equation}\label{eq:gm:powerlaw-direct-q}
 \mathfrak q_{AB}\le C_\beta|A|(1+r)^{-\kappa_\beta},
 \qquad \kappa_\beta=\frac{\pi(\zeta-D)}{\pi+\beta v}>0.
\end{equation}
Here $C$ and $C_{D,\zeta}>0$ depend only on $D,\zeta$, while $C_\beta$ depends only on $\beta,D,J_\zeta,\zeta$.  All constants are independent of the regions and finite volume.
\end{lemma}

\begin{proof}
Set $F(r)=(1+r)^{-\zeta}$.  Since $\zeta>D$, the lattice sum $\sum_z F(|z|_1)$ is finite, and
\begin{equation}\label{eq:gm:powerlaw-convolution}
 \sum_{z\in\Z^D}F(d(x,z))F(d(z,y))
 \le 2^{\zeta+1}\left(\sum_{z\in\Z^D}F(|z|_1)\right)F(d(x,y)).
\end{equation}
Indeed, for each $z$, either $d(x,z)\ge d(x,y)/2$ or $d(z,y)\ge d(x,y)/2$; bound the corresponding factor by $2^\zeta F(d(x,y))$ and sum the other factor.
Thus the Lieb--Robinson bound \cite[Theorem~3.1]{NachtergaeleSimsYoung2019} applies with the interaction bound $J_\zeta$ from \eqref{eq:gm:powerlaw-assumption}.  Followed by Haar averaging on $C$, it gives for an operator $O_X$ supported on $X$ disjoint from $C$,
\begin{equation}\label{eq:gm:powerlaw-local-observable-leakage}
 \|e^{-iuH_0}O_Xe^{iuH_0}
   -\mathsf E_{AB}(e^{-iuH_0}O_Xe^{iuH_0})\|_\infty
 \le C\|O_X\|_\infty e^{v|u|}\sum_{x\in X,c\in C}F(d(x,c)).
\end{equation}
For $X\cap C\ne\varnothing$, the same inequality, with a larger $C$, follows from the trivial bound $2\|O_X\|_\infty$ and the presence of a term $F(0)$ in the sum.
Removing interactions to form $H_0$ preserves the pair bound with the same $J_\zeta$.

Since each cut support meets $A$, \eqref{eq:gm:powerlaw-assumption} and \eqref{eq:gm:powerlaw-convolution} give
\begin{equation}\label{eq:gm:spatial-summation}
 \begin{aligned}
 \sum_{X\text{ crossing }A\mid A^c}\|\Phi(X)\|_\infty
                   \sum_{x\in X,c\in C}F(d(x,c))&\le\sum_{a\in A}\sum_{x\in\Lambda,c\in C}
        F(d(x,c))\sum_{X\ni a,x}\|\Phi(X)\|_\infty\\
 &\le J_\zeta\sum_{a\in A,c\in C}\sum_{x\in\Lambda}F(d(a,x))F(d(x,c))\\
 &\le C_{D,\zeta}J_\zeta\sum_{a\in A,c\in C}F(d(a,c)).
 \end{aligned}
\end{equation}
This treats direct interactions spanning the entire buffer as well as dynamically spreading terms.
Apply the local approximation bound to each term of $V(u)=e^{-iuH_0}Ve^{iuH_0}$ and use \eqref{eq:gm:spatial-summation}.  Then \eqref{eq:gm:duhamel-leakage} gives
\begin{equation}\label{eq:gm:powerlaw-duhamel-leakage-derivation}
 \begin{aligned}
 \eta_{AB}(t)
 &\le2\int_0^{\beta|t|}
 \|V(\operatorname{sgn}(t)u)-\mathsf E_{AB}(V(\operatorname{sgn}(t)u))\|_\infty\,du\\
 &\le C\left(\int_0^{\beta|t|}e^{vu}\,du\right)
 \sum_{X\text{ crossing }A\mid A^c}\|\Phi(X)\|_\infty
 \sum_{x\in X,c\in C}F(d(x,c))\\
 &\le C\beta J_\zeta|t|e^{v\beta|t|}
 \sum_{a\in A,c\in C}F(d(a,c)).
 \end{aligned}
\end{equation}
The last step uses $e^{vu}\le e^{v\beta|t|}$ throughout the integration interval.  Combining this with $\eta_{AB}(t)\le2$ and $F(r)=(1+r)^{-\zeta}$ proves \eqref{eq:gm:F-eta}.

For every finite $A$,
\begin{equation}\label{eq:gm:powerlaw-separated-pair-sum}
 \sum_{a\in A,c\in C}(1+d(a,c))^{-\zeta}
 \le C_{D,\zeta}|A|(1+r)^{D-\zeta}.
\end{equation}
To get a estimate for $\mathfrak q_{AB}$, for $T\ge1$, split the integral defining $\mathfrak q_{AB}$ in \eqref{eq:gm:q-def} at $|t|=T$.
On $|t|\le T$, use the spatial bound above and $e^{v\beta|t|}\le e^{v\beta T}$; on $|t|>T$, use $\eta_{AB}(t)\le2$.  This gives
\begin{equation}\label{eq:gm:powerlaw-modular-time-split}
 \begin{aligned}
 \mathfrak q_{AB}
 &=\frac12\int_{|t|\le T}\frac{\eta_{AB}(t)}{|\sinh(\pi t)|}\,dt
   +\frac12\int_{|t|>T}\frac{\eta_{AB}(t)}{|\sinh(\pi t)|}\,dt\\
 &\le C_\beta|A|(1+r)^{-(\zeta-D)}e^{v\beta T}
       \int_{|t|\le T}\frac{|t|}{|\sinh(\pi t)|}\,dt
       +\int_{|t|>T}\frac{dt}{|\sinh(\pi t)|}\\
 &\le C_\beta|A|(1+r)^{-(\zeta-D)}e^{v\beta T}+Ce^{-\pi T}.
 \end{aligned}
\end{equation}
The last inequality uses $\int_\R |t|/|\sinh(\pi t)|\,dt=1/2$ and \eqref{eq:gm:kernel-tail}.
Choose
\begin{equation}\label{eq:gm:powerlaw-modular-time-cutoff}
 T=\frac{(\zeta-D)\log(1+r)}{\pi+\beta v}.
\end{equation}
Then the two decay factors agree:
\begin{equation}\label{eq:gm:powerlaw-modular-time-balance}
 (1+r)^{-(\zeta-D)}e^{v\beta T}
 =e^{-\pi T}
 =(1+r)^{-\pi(\zeta-D)/(\pi+\beta v)}.
\end{equation}
Since $|A|\ge1$, the preceding bound is therefore at most $C_\beta|A|e^{-\pi T}$.
If the chosen $T$ is less than $1$, \eqref{eq:gm:elementary-leakage} and $g_{A\mid A^c}\le J_\zeta|A|$ instead give
\begin{equation}\label{eq:gm:powerlaw-short-distance-leakage}
 \begin{aligned}
 \mathfrak q_{AB}
 &\le\frac{\beta J_\zeta|A|}{2}\le\frac{e^\pi\beta J_\zeta}{2}|A|e^{-\pi T},
 \end{aligned}
\end{equation}
because $e^{-\pi T}\ge e^{-\pi}$.  Enlarging $C_\beta$ proves \eqref{eq:gm:powerlaw-direct-q} for all $r$.
\end{proof}

\subsection{Completion of the power-law proof}

\begin{proof}[Proof of Theorem~\ref{thm:gm:powerlaw}]
Throughout this proof, constants depend only on $\beta,D,J_\zeta,\zeta$ and the fixed on-site dimension bound.

\emph{Common input: the modular cutoff estimate.}
For the cut pair $\rho=e^{-\beta(H_0+V)}/Z$ and $\sigma=e^{-\beta H_0}/Z_0$, Lemmas~\ref{lem:gm:cut} and~\ref{lem:gm:resolvent} give
\begin{equation}\label{eq:gm:powerlaw-cutoff}
 I(A:C\mid B)_\rho
 \le\left\langle\rho^{1/2},\log(1+e^{-\ell-K})\rho^{1/2}\right\rangle+e^{-\ell}
       +\mathfrak q_{AB}^{\,2}(e^\ell+2\ell),
\end{equation}
where $K=\log(L_\sigma R_{\rho^{-1}})$.  Lemma~\ref{lem:gm:powerlaw-modular-tail} implies, for $\ell\ge1$,
\begin{equation}\label{eq:gm:powerlaw-tail-input}
 \left\langle\rho^{1/2},\log(1+e^{-\ell-K})\rho^{1/2}\right\rangle
 \le C_\beta (1+|A|)^{b_\zeta+1}\ell^{-b_\zeta}
 \begin{cases}
 1,&\zeta/D\notin\mathbb N,\\
 [\log(e+\ell)]^2,&\zeta/D\in\mathbb N.
 \end{cases}
\end{equation}
Its proof tracks the region factor through $m=\lceil\zeta/D\rceil$ modular actions and the support-tail exponent $s=1+\zeta/D-m$: squaring gives $(1+|A|)^{2(m+s)}=(1+|A|)^{b_\zeta+1}$.  The estimate is uniform under deletion of interaction terms, so it also applies to the truncated Gibbs pairs below.

\emph{Case 1: $D<\zeta\le2D$.}
Choose $\ell=\kappa_\beta\log(1+r)$, with $\kappa_\beta$ from \eqref{eq:gm:powerlaw-direct-q}.  First suppose $D<\zeta<2D$.  Then $1<\zeta/D<2$, so the spectral estimate has no endpoint logarithm.  For sufficiently large $r$, \eqref{eq:gm:powerlaw-cutoff}--\eqref{eq:gm:powerlaw-tail-input} give
\begin{equation}\label{eq:gm:powerlaw-inverse-log-completion}
 \begin{aligned}
 I(A:C\mid B)_\rho
 &\le C_\beta(1+|A|)^{b_\zeta+1}[\log(1+r)]^{-b_\zeta}
       +C_\beta(1+|A|)^2(1+r)^{-\kappa_\beta}\\
 &\le C_\beta(1+|A|)^{b_\zeta+1}[\log(e+r)]^{-b_\zeta}.
 \end{aligned}
\end{equation}
The second term is absorbed because an inverse power of $r$ decays faster than an inverse power of $\log r$.  This proves \eqref{eq:gm:powerlaw-log-main} for large $r$.

At $\zeta=2D$, we have $b_\zeta=5$ and $\zeta/D=2\in\mathbb N$, so we use the second branch of \eqref{eq:gm:powerlaw-tail-input}, which includes the factor $[\log(e+\ell)]^2$.  With the same cutoff, this gives
\begin{equation}\label{eq:gm:powerlaw-log-endpoint-completion}
 \begin{aligned}
 I(A:C\mid B)_\rho
 &\le C_\beta(1+|A|)^6\ell^{-5}[\log(e+\ell)]^2
       +C_\beta(1+|A|)^2(1+r)^{-\kappa_\beta}\\
 &\le C_\beta(1+|A|)^6[\log(e+r)]^{-5}
       [\log(e+\log(e+r))]^2.
 \end{aligned}
\end{equation}
Here $\ell=\kappa_\beta\log(1+r)$, and the algebraically decaying term is again absorbed into the inverse-logarithmic bound.  This proves \eqref{eq:gm:powerlaw-log-endpoint} for large $r$.

\emph{Case 2: $\zeta>2D$.}
We first truncate long interactions touching the shielding neighborhood around $A$.  For sufficiently large $r$, choose an integer $1\le L\le r/24$ and set
\begin{equation}\label{eq:gm:powerlaw-collar}
 \Omega=\{x\in\Lambda:d(x,A)\le\lfloor r/3\rfloor\},
 \qquad
 H^{(L)}=H_\Lambda-
 \sum_{\substack{X\subseteq\Lambda\\X\cap\Omega\ne\varnothing\\\diam X>L}}\Phi(X).
\end{equation}
Then $\Omega\subseteq AB$ and $|\Omega|\le C_D(1+|A|)(1+r)^D$.  Lemma~\ref{lem:gm:powerlaw-support-tail} gives
\begin{equation}\label{eq:gm:powerlaw-truncation-norm}
 \|H_\Lambda-H^{(L)}\|_\infty
 \le C J_\zeta(1+|A|)(1+r)^D(1+L)^{D-\zeta}.
\end{equation}
Let $\rho^{(L)}$ be the Gibbs state of $H^{(L)}$.  Lemma~\ref{lem:gm:gibbs-cmi-lipschitz} and $\log d_A\le C|A|$, where $d_A=\dim\mathcal H_A$, imply
\begin{equation}\label{eq:gm:powerlaw-truncation-cmi}
 \bigl|I(A:C\mid B)_\rho-I(A:C\mid B)_{\rho^{(L)}}\bigr|
 \le C_\beta (1+|A|)^2(1+r)^D(1+L)^{D-\zeta}.
\end{equation}

For propagation inside the neighborhood, choose the stronger spatial weight
\begin{equation}\label{eq:gm:powerlaw-collar-weight}
 \begin{aligned}
 w_L&=\frac{\log(1+L)}{L}, \qquad
 \sup_x\sum_y e^{w_Ld(x,y)}
 \sum_{\substack{X\ni x,y\\X\subseteq\Omega\\\diam X\le L}}\|\Phi(X)\|_\infty
 &\le J_\zeta\sum_{z\in\Z^D}(1+|z|_1)^{1-\zeta}<\infty.
 \end{aligned}
\end{equation}
Indeed, convexity gives $e^{w_Lu}\le1+u$ for $0\le u\le L$, and $\zeta>2D$ implies $\zeta>D+1$.  The bound is independent of $L$.
Applying \eqref{eq:gm:powerlaw-collar-q} from Lemma~\ref{lem:gm:powerlaw-collar-leakage} with $w=w_L=\log(1+L)/L$ yields
\begin{equation}\label{eq:gm:powerlaw-weighted-collar-q}
 \mathfrak q_{AB}^{(L)}
 \le C_\beta (1+|A|)^2(1+r)^D
 \exp\!\left[-c_\beta\frac{r\log(1+L)}L\right].
\end{equation}

Choose the range cutoff, with integer rounding, as
\begin{equation}\label{eq:gm:powerlaw-optimal-range}
 L\asymp r^\alpha,
 \qquad \alpha=\frac{D+b_\zeta}{\zeta-D+b_\zeta}\in(0,1),
\end{equation}
and choose $\ell=c_0r\log(1+L)/L$, with $c_0>0$ sufficiently small relative to the decay constant in \eqref{eq:gm:powerlaw-weighted-collar-q}.  Then $L\le r/24$ for sufficiently large $r$, with a threshold independent of $A$, and
\begin{equation}\label{eq:gm:powerlaw-algebraic-cutoff-scales}
 \ell\asymp r^{1-\alpha}\log r,
 \qquad \log(e+\ell)\asymp\log r.
\end{equation}
Let $K^{(L)}$ be the relative modular logarithm of the truncated pair.  Applying the uniformly valid logarithmic upper bound in \eqref{eq:gm:powerlaw-tail-input} to this pair gives
\begin{equation}\label{eq:gm:powerlaw-log-absorption}
 \begin{aligned}
 \left\langle(\rho^{(L)})^{1/2},\log(1+e^{-\ell-K^{(L)}})(\rho^{(L)})^{1/2}\right\rangle &\le C_\beta (1+|A|)^{b_\zeta+1}\ell^{-b_\zeta}[\log(e+\ell)]^2\\
 &\le C_\beta (1+|A|)^{b_\zeta+1}r^{-b_\zeta(1-\alpha)}(\log r)^{2-b_\zeta}\\
 &\le C_\beta (1+|A|)^{b_\zeta+1}r^{-b_\zeta(1-\alpha)},
 \end{aligned}
\end{equation}
since $b_\zeta>5$.  Thus the logarithmic gain in the localization estimate absorbs the spectral endpoint logarithm.
The remaining terms in \eqref{eq:gm:powerlaw-cutoff} are bounded by
\[
 C_\beta (1+|A|)^4(1+r)^{2D+1}e^{-c_\beta r^{1-\alpha}\log r},
\]
which decays faster than every inverse power of $r$, and $(1+|A|)^4\le (1+|A|)^{b_\zeta+1}$.  Combining with \eqref{eq:gm:powerlaw-truncation-cmi} gives
\begin{equation}\label{eq:gm:powerlaw-two-scales}
 I(A:C\mid B)_\rho
 \le C_\beta (1+|A|)^{b_\zeta+1}
 \left[r^{-b_\zeta(1-\alpha)}+r^{D-\alpha(\zeta-D)}\right].
\end{equation}
The choice of $\alpha$ balances the two exponents:
\begin{equation}\label{eq:gm:powerlaw-algebraic-exponent-balance}
 b_\zeta(1-\alpha)=\alpha(\zeta-D)-D
 =\frac{b_\zeta(\zeta-2D)}{\zeta-D+b_\zeta}=\nu_\zeta.
\end{equation}
This proves \eqref{eq:gm:powerlaw-main} for sufficiently large $r$, including integer values of $\zeta/D$.

For bounded $r$, all three estimates follow by enlarging $C_\beta$, using the same elementary CMI bound and prefactor adjustment as in \eqref{eq:gm:quasilocal-bounded-distance}. 
This completes both cases. 
In summaary, the coefficient estimates for $O_j$ produce the polynomial dependence on $|A|$. The direct leakage bound gives inverse-logarithmic decay for $D<\zeta\le2D$, while balancing truncation and spectral errors determines the algebraic exponent $\nu_\zeta$ for $\zeta>2D$.

\end{proof}

\section{Summary and Discussion}\label{sec:gm:discussion}

Our results relate approximate conditional independence in Gibbs states to the strength and spatial decay of their interactions, without assuming clustering or rapid mixing. 
For finite-range interactions, the exponential prefactor is controlled by the interaction strength across the cut and hence by the boundary size.  
This improves the exponential volume dependence of previous local Markov bounds while retaining exponential decay with separation at every fixed positive temperature.  
The size dependence involves only region $A$ or, by symmetry, region $C$. For exponentially and stretched-exponentially decaying interactions, the static comparison with a cut Gibbs state remains effective even when interaction supports have unbounded size.
Truncating long interactions near $A$ and combining spatial localization with spectral estimates gives the distance exponent $\theta/(D+1-\theta)$ and a boundary-dependent prefactor.

In two dimensions, both the finite-range and exponentially decaying interaction bounds have a useful consequence for coarse-graining.  Consider regular full partitions $\Lambda_s=A_s\sqcup B_s\sqcup C_s$ obtained by enlarging a fixed geometry, with
\begin{equation}\label{eq:gm:planar-rescaling-geometry}
 |\partial_{\mathrm e}A_s|\le ps,
 \qquad d(A_s,C_s)\ge qs,
\end{equation}
where $s$ is the length scale and $p,q>0$ are fixed geometric constants.  The finite-range bound gives
\begin{equation}\label{eq:gm:discussion-2d-scaling}
 I(A_s:C_s\mid B_s)_{\rho_{\Lambda_s,\beta}}
 \le C_\beta\exp\!\left[-(c_\beta q-C_\beta p)s\right].
\end{equation}
At fixed $\beta$, choosing a sufficiently wide buffer relative to the size of $A_s$ makes $c_\beta q>C_\beta p$, so the CMI tends to zero exponentially as the entire partition is enlarged with these proportions held fixed.
The essential point in both cases is that the boundary length and buffer width scale linearly in two dimensions.
This suppression under geometric rescaling has a renormalization-group interpretation and may find applications in approximate entanglement bootstrap for mixed quantum states in two dimensions \cite{ShiKatoKim2020, yang2025topologicalmixedstatesphases}.
This may also improve the characterization of thermal and mixed-state phases through local quantum channels and local reversibility \cite{Ma2025Gibbs,Sang2024,sang2025mixedstatephaseslocalreversibility,hsieh2026locallystable,chen2026notestrongquantummarkov}.

For power-law interactions, the bounds have polynomial region-size dependence: Theorem~\ref{thm:gm:powerlaw} gives inverse-logarithmic decay for $D<\zeta\le2D$, with a squared iterated-logarithm correction at $\zeta=2D$, and algebraic decay for $\zeta>2D$. 
The sufficient threshold $2D$ arises from the truncation error and the requirement that the retained interaction range be smaller than the separation.
Obtaining algebraic decay for unrestricted-support interactions throughout $D<\zeta\le2D$, improving the distance exponents, and establishing stronger boundary dependence remain open. 

Substantial progress has been made in quantum Gibbs sampling, including efficiently implementable detailed-balanced dynamics 
\cite{KastoryanoBrandao2016Samplers,ChenKastoryanoBrandaoGilyen2023,ChenKastoryanoGilyen2023,DingLiLin2025,GilyenChenDoriguelloKastoryano2024}, and in bounding their mixing times 
\cite{KastoryanoTemme2013Mixing,BardetEtAl2023Thermalization,KochanowskiEtAl2025Thermalization,CapelEtAl2025Sampling,RouzeFrancaAlhambra2026Thermalization,RouzeFrancaAlhambra2026Optimal,BakshiLiuMoitraTang2025,TongZhan2025Fermions,SmidEtAl2025FermiHubbard,SmidEtAl2025WeakMixing,BergamaschiChen2025SpinChains,Bergamaschi2026AllToAll}.
This raises the question of whether the Lindbladian approach to local Markovianity \cite{Chen2025} can also recover our exponential boundary dependence, with dependence on the boundary of $A$ but no size dependence on $C$ or the total system.
Under suitable uniform rapid-mixing assumptions, can it further improve the current polynomial dependence on the total region size to only boundary region size while retaining exponential decay with separation?  The static comparison approach used here is not directly applicable to the Lindbladian setting, but it may provide a useful starting point for further analysis.

Moreover, our finite-range bound depends on the strength of the boundary interactions, with no explicit dependence on system size or local Hilbert-space dimension, suggesting a promising starting point for extensions to bosonic systems.
Existing results include Lieb--Robinson bounds for interacting bosons under suitable occupation assumptions \cite{KuwaharaSaito2021Bosons,KuwaharaVuSaito2024Bosons}, as well as high-temperature clustering, local particle-number moment bounds, and a thermal area law for Bose--Hubbard type Gibbs states \cite{TongKuwaharaGong2026Bosons}.
The main obstacle is to control the unbounded boundary interactions uniformly as the local occupation cutoff is removed.
A possible route is to combine thermal bounds on local occupation numbers with spatial localization and modular spectral estimates, 
while controlling entropy errors through energy-constrained continuity bounds \cite{Winter_2016}.

A related question is whether the exponential boundary dependence extends to metastable states.
Bergamaschi, Chen, and Vazirani \cite{bergamaschi2025structuraltheoryquantummetastability} study states satisfying $\|\mathcal L(\sigma)\|_1\le\epsilon$ for their quasi-local KMS detailed-balanced thermal Lindbladian and prove a local recovery property and a mutual-information area law with metastability corrections.
Their recovery error is bounded by $n e^{\mu|A|}\epsilon^\lambda$ and the correction to the area law is $n^2 e^{\mu|A|}\epsilon^\lambda$, where $\mu,\lambda>0$ depend on temperature and interaction parameters.
Thus these estimates retain exponential region-volume dependence; the area law itself does not establish a boundary-exponential prefactor for CMI decay.
Our static comparison uses the exact Gibbs form and does not directly apply to an approximately stationary state, which may remain far from equilibrium.
A natural question is whether, for finite-range lattice systems, the exponential dependence can be confined to $|\partial_{\mathrm e}A|$, while allowing polynomial region- or system-size factors in the metastability correction.
Such an extension would require additional control of the departure from local equilibrium and is not established by the present Gibbs-state argument.

\begin{acknowledgments}
T.H.Y. thanks Anthony Chen for useful discussions. T.H.Y is supported by Taiwan-UIUC fellowship program.
\end{acknowledgments}

\appendix
\section{Duhamel's formula and Gr\"onwall's inequality}\label{app:gm:gronwall}

\begin{lemma}[Duhamel's formula]\label{lem:gm:duhamel}
Let $G_1(t),G_2(t)$ be continuous Hermitian matrix-valued functions, and let $U_j(t,s)$ solve $\partial_tU_j(t,s)=-iG_j(t)U_j(t,s)$ with $U_j(s,s)=\one$.  For $t\ge0$,
\begin{equation}\label{eq:gm:duhamel-propagator}
 U_1(t,0)-U_2(t,0)
 =-i\int_0^t U_1(t,s)\bigl(G_1(s)-G_2(s)\bigr)U_2(s,0)\,ds.
\end{equation}
Consequently, unitarity gives
\begin{equation}\label{eq:gm:duhamel-propagator-norm}
 \|U_1(t,0)-U_2(t,0)\|_\infty
 \le\int_0^t\|G_1(s)-G_2(s)\|_\infty\,ds.
\end{equation}
For a continuously differentiable matrix-valued function $M(s)$, the exponential derivative takes the form
\begin{equation}\label{eq:gm:duhamel-exponential-derivative}
 \frac{d}{ds}e^{M(s)}
 =\int_0^1 e^{(1-u)M(s)}M'(s)e^{uM(s)}\,du.
\end{equation}
In particular, for $\rho_s=e^{-\beta(H+sW)}/\Tr e^{-\beta(H+sW)}$ with $H,W$ Hermitian,
\begin{equation}\label{eq:gm:duhamel-normalized-gibbs}
 \rho_s'=-\beta\int_0^1\rho_s^{1-u}
 \bigl(W-\Tr(\rho_sW)\one\bigr)\rho_s^u\,du.
\end{equation}
\end{lemma}

\begin{proof}
Differentiate $U_1(t,s)U_2(s,0)$ with respect to $s$.  Its derivative is $iU_1(t,s)(G_1(s)-G_2(s))U_2(s,0)$.  Integrating from $0$ to $t$ proves \eqref{eq:gm:duhamel-propagator}; taking the norm proves \eqref{eq:gm:duhamel-propagator-norm}.
For the exponential derivative, differentiate $e^{(1-u)M(s+h)}e^{uM(s)}$ with respect to $u$ and integrate to obtain
\begin{equation}\label{eq:gm:duhamel-exponential-difference}
 e^{M(s+h)}-e^{M(s)}
 =\int_0^1e^{(1-u)M(s+h)}\bigl(M(s+h)-M(s)\bigr)e^{uM(s)}\,du.
\end{equation}
Divide by $h$ and let $h\to0$ to prove \eqref{eq:gm:duhamel-exponential-derivative}.  Apply it with $M(s)=-\beta(H+sW)$ and differentiate the normalization.  Cyclicity gives $\frac{d}{ds}\log\Tr e^{-\beta(H+sW)}=-\beta\Tr(\rho_sW)$, yielding \eqref{eq:gm:duhamel-normalized-gibbs}.
\end{proof}

\begin{lemma}[Gr\"onwall's inequality]\label{lem:gm:gronwall}
Let $T>0$, let $f:[0,T]\to[0,\infty)$ be continuous, and let $a:[0,T]\to[0,\infty)$ be integrable.  If, for some $c\ge0$,
\begin{equation}\label{eq:gm:gronwall-integral-assumption}
 f(u)\le c+\int_0^u a(s)f(s)\,ds
 \qquad(0\le u\le T),
\end{equation}
then
\begin{equation}\label{eq:gm:gronwall-exponential-bound}
 f(u)\le c\exp\!\left(\int_0^u a(s)\,ds\right)
 \qquad(0\le u\le T).
\end{equation}
\end{lemma}

\begin{proof}
Put $g(u)=c+\int_0^u a(s)f(s)\,ds$.  Then $f\le g$, $g(0)=c$, and
$g'(u)=a(u)f(u)\le a(u)g(u)$ almost everywhere.
Multiplying by $\exp[-\int_0^u a(s)\,ds]$ and integrating proves the assertion.
For $c=0$, first replace $c$ by $\varepsilon>0$ and then let $\varepsilon\downarrow0$.
\end{proof}

\section{Resolvent representation of the entropy loss}\label{app:gm:gap-integral}

\noindent\textbf{Lemma~\ref{lem:gm:gap-integral} (restated).} Let $\rho,\sigma$ be faithful states on the finite-dimensional tensor product $R\otimes E$, and set $\Delta_R=L_{\sigma_R}R_{\rho_R^{-1}}$.  Set $\Delta=L_\sigma R_{\rho^{-1}}$ and $\delta_R(\rho,\sigma)=D(\rho\|\sigma)-D(\rho_R\|\sigma_R)$.  Then
\[
 \begin{split}
 \delta_R(\rho,\sigma)&=\int_0^\infty j(s)\,ds,\\
 j(s)&=\langle\rho^{1/2},(\Delta+s)^{-1}\rho^{1/2}\rangle
 -\langle\rho_R^{1/2},(\Delta_R+s)^{-1}\rho_R^{1/2}\rangle.
 \end{split}
\]

\begin{proof}
Left multiplication by $\sigma$ and right multiplication by $\rho^{-1}$ are commuting positive operators on Hilbert--Schmidt space.  Hence
\begin{equation}\label{eq:gm:appendix-modular-log-action}
 (\log\Delta)(Y)=(\log\sigma)Y-Y\log\rho,
\end{equation}
and therefore
\begin{equation}\label{eq:gm:appendix-modular-relative-entropy}
 -\langle\rho^{1/2},\log\Delta\,\rho^{1/2}\rangle
 =\Tr\rho(\log\rho-\log\sigma)=D(\rho\|\sigma).
\end{equation}
The same calculation on the reduced space gives
$-\langle\rho_R^{1/2},\log\Delta_R\,\rho_R^{1/2}\rangle=D(\rho_R\|\sigma_R)$.

For $\lambda>0$, direct integration gives
\begin{equation}\label{eq:gm:scalar-log-resolvent-integral}
 \int_0^T\left(\frac1{\lambda+s}-\frac1{1+s}\right)ds
 =\log\frac{\lambda+T}{1+T}-\log\lambda
 \xrightarrow[T\to\infty]{}-\log\lambda.
\end{equation}
Applying this identity by spectral calculus and using
$\|\rho^{1/2}\|_2^2=\|\rho_R^{1/2}\|_2^2=1$, we obtain
\begin{equation}\label{eq:gm:relative-entropy-resolvent-integrals}
 \begin{aligned}
 D(\rho\|\sigma)
 &=\int_0^\infty\left(\langle\rho^{1/2},(\Delta+s)^{-1}\rho^{1/2}\rangle-\frac1{1+s}\right)ds,\\
 D(\rho_R\|\sigma_R)
 &=\int_0^\infty\left(\langle\rho_R^{1/2},(\Delta_R+s)^{-1}\rho_R^{1/2}\rangle-\frac1{1+s}\right)ds.
 \end{aligned}
\end{equation}
Both integrals converge absolutely. Subtracting them cancels the scalar terms and proves \eqref{eq:gm:gap-integral}.
\end{proof}

\subsection{Derivation of the Fourier resolvent identity}\label{app:gm:fourier-resolvent}
For completeness, we derive \eqref{eq:gm:fourier-resolvent} from the standard integral \cite[Eq.~(4.40.8)]{DLMF}
\begin{equation}\label{eq:gm:hyperbolic-sine-transform}
 \int_0^\infty\frac{\sinh(at)}{\sinh(\pi t)}\,dt
 =\frac12\tan(a/2),\qquad -\pi<a<\pi.
\end{equation}
Both sides are holomorphic in the strip $|\operatorname{Re}a|<\pi$, so the identity extends to that strip.  Substituting $a=iv$, with $v\in\R$, and cancelling $i$ gives
\begin{equation}\label{eq:gm:sine-hyperbolic-transform}
 \int_0^\infty\frac{\sin(vt)}{\sinh(\pi t)}\,dt
 =\frac12\tanh(v/2).
\end{equation}
Set $v=u-\log s$ and $w=-\log s$.  Since
$s^{-it}(e^{itu}-1)=e^{ivt}-e^{iwt}$, pairing the positive and negative halves of the integral yields
\begin{equation}\label{eq:gm:resolvent-fourier-transform-calculation}
 \begin{aligned}
 \int_\R\frac{s^{-it}(e^{itu}-1)}{\sinh(\pi t)}\,dt
 &=2i\int_0^\infty\frac{\sin(vt)-\sin(wt)}{\sinh(\pi t)}\,dt\\
 &=i\left[\tanh\!\left(\frac{u-\log s}{2}\right)
          -\tanh\!\left(\frac{-\log s}{2}\right)\right]\\
 &=i\left[\frac{e^u-s}{e^u+s}-\frac{1-s}{1+s}\right].
 \end{aligned}
\end{equation}
Multiplication by $i/(2s)$ therefore gives
\begin{equation}\label{eq:gm:resolvent-fourier-identity-completion}
 \frac{i}{2s}\int_\R\frac{s^{-it}(e^{itu}-1)}{\sinh(\pi t)}\,dt
 =\frac1{s+e^u}-\frac1{1+s},
\end{equation}
which is \eqref{eq:gm:fourier-resolvent}.  All integrals in the last two displays are ordinary absolutely convergent integrals: the numerator vanishes to first order at $t=0$, and $1/\sinh(\pi t)$ decays exponentially at infinity.

\section{Conditional logarithm second-moment bound}\label{app:gm:conditional-logarithm}

\begin{lemma}[Conditional logarithm second-moment bound]\label{lem:gm:conditional-logarithm}
Let $\omega_{AY}$ be a faithful state and $d_A=\dim\mathcal H_A$.  Define
\begin{equation}\label{eq:gm:conditional-logarithm-reference}
 \tau=\frac{\one_A}{d_A}\otimes\omega_Y,
 \qquad R_{A\mid Y}=\log\omega_{AY}-\log\tau.
\end{equation}
Then
\begin{equation}\label{eq:gm:conditional-log-variance}
 \Tr(\omega_{AY}R_{A\mid Y}^2)
 \le(e-1)(1+2\log d_A)^2.
\end{equation}
\end{lemma}

\begin{proof}
To show $\omega_{AY}\le d_A\one_A\otimes\omega_Y$, first consider a pure state with Schmidt decomposition $|\psi\rangle=\sum_{i=1}^k\sqrt{\lambda_i}|a_i\rangle\otimes|y_i\rangle$, where $k\le d_A$ and $\psi_Y=\sum_i\lambda_i|y_i\rangle\langle y_i|$.
For every vector $|v\rangle$, Cauchy--Schwarz gives
\begin{equation}\label{eq:gm:conditional-pure-state-bound}
 |\langle\psi|v\rangle|^2
 \le k\sum_{i=1}^k\lambda_i|\langle a_i\otimes y_i|v\rangle|^2
 \le d_A\langle v|(\one_A\otimes\psi_Y)|v\rangle.
\end{equation}
Thus $|\psi\rangle\langle\psi|\le d_A\one_A\otimes\psi_Y$.
Applying this inequality to each pure state in a convex decomposition of $\omega_{AY}$ and summing gives $\omega_{AY}\le d_A\one_A\otimes\omega_Y=d_A^2\tau$.
Inversion reverses the operator order, giving $\tau^{-1}\le d_A^2\omega_{AY}^{-1}$; multiplying on both sides by $\omega_{AY}$ and taking the trace yields $\Tr(\omega_{AY}^2\tau^{-1})\le d_A^2\Tr\omega_{AY}=d_A^2$.
We claim that
\begin{equation}\label{eq:gm:conditional-surprisal-mgf}
 \Tr(\omega_{AY}^{1+\alpha}\tau^{-\alpha})\le d_A^{2\alpha},
 \qquad
 \Tr(\omega_{AY}^{1-\alpha}\tau^\alpha)\le1,
 \qquad 0\le\alpha\le1,
\end{equation}
For the first inequality, the function $\alpha\mapsto\Tr(\omega_{AY}^{1+\alpha}\tau^{-\alpha})$ is log-convex. Interpolation between $\alpha=0$ and $\alpha=1$ gives
\begin{equation}\label{eq:gm:conditional-moment-interpolation}
 \begin{aligned}
 \Tr(\omega_{AY}^{1+\alpha}\tau^{-\alpha})
 &\le (\Tr\omega_{AY})^{1-\alpha}
       [\Tr(\omega_{AY}^2\tau^{-1})]^\alpha
 \le d_A^{2\alpha}.
 \end{aligned}
\end{equation}
For the second inequality, apply Schatten H\"older with conjugate exponents $1/(1-\alpha)$ and $1/\alpha$, for $0<\alpha<1$:
\begin{equation}\label{eq:gm:conditional-moment-holder}
 \Tr(\omega_{AY}^{1-\alpha}\tau^\alpha)
 \le (\Tr\omega_{AY})^{1-\alpha}(\Tr\tau)^\alpha=1.
\end{equation}

To use these estimates without assuming commutativity, diagonalize
$\omega_{AY}=\sum_i p_i|i\rangle\langle i|$ and
$\tau=\sum_j q_j|j\rangle\langle j|$.
The probabilities $p_i|\langle i|j\rangle|^2$ define a scalar random variable
$u_{ij}=\log p_i-\log q_j$.  Its second moment equals
$\Tr(\omega_{AY}R_{A\mid Y}^2)$, and its exponential moments are exactly the traces in \eqref{eq:gm:conditional-surprisal-mgf}.
The scalar inequality
$u^2\le\alpha^{-2}(e^{\alpha u}+e^{-\alpha u}-2)$ therefore gives, for $0<\alpha\le1$,
\begin{equation}\label{eq:gm:conditional-log-second-moment-alpha}
 \begin{aligned}
 \mathbb E[u^2] = \Tr(\omega_{AY}R_{A\mid Y}^2)
 &\le\frac{\mathbb E[e^{\alpha u}]+\mathbb E[e^{-\alpha u}]-2}{\alpha^2}\\
 &=\frac{\Tr(\omega_{AY}^{1+\alpha}\tau^{-\alpha})
       +\Tr(\omega_{AY}^{1-\alpha}\tau^\alpha)-2}{\alpha^2}\\
 &\le\frac{d_A^{2\alpha}+1-2}{\alpha^2}
 =\frac{d_A^{2\alpha}-1}{\alpha^2}.
 \end{aligned}
\end{equation}
Choosing $\alpha=(1+2\log d_A)^{-1}$ gives $d_A^{2\alpha}\le e$ and proves \eqref{eq:gm:conditional-log-variance}.
\end{proof}

\clearpage
\section{Notation}\label{app:gm:notation}

Table~\ref{tab:gm:notation} collects the principal notation.  The references specify the scope of symbols used in different parts of the proof; auxiliary integration variables are defined locally.

\begingroup
\small
\setlength{\tabcolsep}{4pt}
\renewcommand{\arraystretch}{1.15}
\begin{center}
\makeatletter
\def\@captype{table}
\makeatother
\caption{Principal notation and its defining locations.}
\label{tab:gm:notation}
\begin{tabular}{p{3.1cm}p{9.4cm}p{2.8cm}}
\hline
\textbf{Symbol} & \textbf{Meaning} & \textbf{Definition} \tabularnewline
\hline
\multicolumn{3}{l}{\textbf{Geometry, states, and interactions}} \tabularnewline
\hline
\raggedright $\Gamma,\Lambda$ & \raggedright Underlying set of sites and a finite volume $\Lambda\Subset\Gamma$. & \raggedright Sec.~\ref{sec:gm:setting} \tabularnewline
\raggedright $A,B,C;\ r$ & \raggedright Full partition $\Lambda=A\sqcup B\sqcup C$; separation $r=d(A,C)$. The complement $A^c$ is taken in $\Lambda$. & \raggedright \eqref{eq:gm:full-partition} \tabularnewline
\raggedright $D;\ \partial_{\mathrm e}A$ & \raggedright Spatial dimension; nearest-neighbor edges with exactly one endpoint in $A$. & \raggedright Thm.~\ref{thm:gm:local} \tabularnewline
\raggedright $\mathcal H_X;\ d_A$ & \raggedright Tensor-product Hilbert space on $X$; $d_A=\dim\mathcal H_A$. & \raggedright Sec.~\ref{sec:gm:setting}; Lem.~\ref{lem:gm:gibbs-cmi-lipschitz} \tabularnewline
\raggedright $\Phi(X),h_\gamma$ & \raggedright Interaction terms indexed by their support $X$, or individually by $\gamma$. & \raggedright \eqref{eq:gm:gibbs}; Thm.~\ref{thm:gm:local} \tabularnewline
\raggedright $H_\Lambda,\beta,Z_\Lambda$ & \raggedright Hamiltonian, inverse temperature, and partition function; $\rho_{\Lambda,\beta}=e^{-\beta H_\Lambda}/Z_\Lambda$. & \raggedright \eqref{eq:gm:gibbs} \tabularnewline
\raggedright $H_0,V;\ g_{A\mid A^c}$ & \raggedright $H_0=H_A+H_{BC}$ and $H_\Lambda=H_0+V$; $g_{A\mid A^c}$ sums the norms of all interaction terms crossing $A\mid A^c$. & \raggedright \eqref{eq:gm:cut-strength}; Lem.~\ref{lem:gm:cut} \tabularnewline
\raggedright $\rho,\sigma;\ \rho_X$ & \raggedright Faithful states in the modular argument; for the cut Gibbs pair, $\rho=e^{-\beta H_\Lambda}/Z$ and $\sigma=e^{-\beta H_0}/Z_0$. A subscript $X$ denotes a reduced state. & \raggedright Lem.~\ref{lem:gm:cut} \tabularnewline
\raggedright $S(\omega),I(A:C\mid B)_\rho$ & \raggedright Von Neumann entropy and conditional mutual information, with natural logarithms. & \raggedright \eqref{eq:gm:intro-cmi} \tabularnewline
\raggedright $D(\rho\|\sigma)$ & \raggedright Quantum relative entropy; distinct from the spatial dimension $D$. & \raggedright \eqref{eq:gm:entropy-definitions} \tabularnewline
\raggedright $\|O\|_\infty,\|O\|_1,\|O\|_2$ & \raggedright Operator, trace, and Hilbert--Schmidt norms; $\langle X,Y\rangle=\Tr(X^\dagger Y)$. & \raggedright Sec.~\ref{sec:gm:setting} \tabularnewline
\hline
\multicolumn{3}{l}{\textbf{Relative modular operators and localization}} \tabularnewline
\hline
\raggedright $R,E$ & \raggedright Retained and traced-out subsystems in the abstract argument; $R=AB$ and $E=C$ for the Gibbs application. & \raggedright Sec.~\ref{sec:gm:modular} \tabularnewline
\raggedright $L_X,R_X$ & \raggedright Left and right multiplication: $L_X(Y)=XY$, $R_X(Y)=YX$. & \raggedright \eqref{eq:gm:relative-modular} \tabularnewline
\raggedright $\Delta,\Delta_R;\ K$ & \raggedright $\Delta=L_\sigma R_{\rho^{-1}}$, $\Delta_R=L_{\sigma_R}R_{\rho_R^{-1}}$; $K=\log\Delta$, acting on Hilbert--Schmidt space. & \raggedright \eqref{eq:gm:relative-modular}; Lem.~\ref{lem:gm:gap-integral} \tabularnewline
\raggedright $\delta_R(\rho,\sigma)$ & \raggedright Relative-entropy loss under $\Tr_E$: $D(\rho\|\sigma)-D(\rho_R\|\sigma_R)$. & \raggedright \eqref{eq:gm:loss} \tabularnewline
\raggedright $j(s)$ & \raggedright Full minus reduced resolvent expectation, whose integral is $\delta_R$; $s>0$ is the resolvent parameter. & \raggedright \eqref{eq:gm:gap-integral} \tabularnewline
\raggedright $\mathsf E_R,\eta_R(t)$ & \raggedright Normalized partial-trace conditional expectation and modular leakage of $\sigma^{it}\rho^{-it}$ outside $R$. & \raggedright \eqref{eq:gm:q-def} \tabularnewline
\raggedright $\mathfrak q_R$ & \raggedright Integrated modular leakage $\frac12\int_\R\eta_R(t)/|\sinh(\pi t)|\,dt$. & \raggedright \eqref{eq:gm:q-def} \tabularnewline
\raggedright $a,b;\ \ell$ & \raggedright Lower and upper resolvent cutoffs, $0<a<1<b$; the symmetric choice is $a=e^{-\ell}$, $b=e^\ell$, with $\ell>0$. & \raggedright Lem.~\ref{lem:gm:resolvent} \tabularnewline
\raggedright $\alpha,\alpha_0$ & \raggedright Negative modular-moment parameter and its uniform admissible upper bound in the finite-range proof. & \raggedright \eqref{eq:gm:alpha0} \tabularnewline
\hline
\end{tabular}
\end{center}
\newpage
\begin{center}
\textbf{Table~\ref{tab:gm:notation} (continued)}\par\medskip
\begin{tabular}{p{3.1cm}p{9.4cm}p{2.8cm}}
\hline
\textbf{Symbol} & \textbf{Meaning} & \textbf{Definition} \tabularnewline
\hline
\multicolumn{3}{l}{\textbf{Spatial decay and support weights}} \tabularnewline
\hline
\raggedright $R_0,J,\mathfrak d$ & \raggedright Finite interaction range, term-norm bound, and maximum number of overlapping terms, counting the term itself. & \raggedright Thm.~\ref{thm:gm:local} \tabularnewline
\raggedright $\mu_{\mathrm{LR}},v$ & \raggedright Spatial decay coefficient and temporal growth rate in the finite-range Lieb--Robinson bound. The rate $v$ is chosen separately for each interaction class. & \raggedright \eqref{eq:gm:local-LR} \tabularnewline
\raggedright $\lambda_\beta$ & \raggedright Finite-range integrated modular leakage decay rate $\pi\mu_{\mathrm{LR}}/(\pi+\beta v)$. & \raggedright \eqref{eq:gm:local-q} \tabularnewline
\raggedright $\mu,\theta,J_{\mu,\theta}$ & \raggedright Spatial decay parameters $\mu>0$, $0<\theta\le1$, and aggregate interaction bound $\sum_{X\ni x,y}\|\Phi(X)\|_\infty\le J_{\mu,\theta}e^{-\mu d(x,y)^\theta}$. & \raggedright \eqref{eq:gm:F-assumption} \tabularnewline
\raggedright $F$ & \raggedright Power-law profile $F(r)=(1+r)^{-\zeta}$ used in the spatial leakage proof. & \raggedright Sec.~\ref{sec:gm:spatial} \tabularnewline
\raggedright $a_{\mathrm{supp}}$ & \raggedright Positive support-weight parameter, distinct from the resolvent cutoff $a$. & \raggedright Lem.~\ref{lem:gm:support-moments} \tabularnewline
\raggedright $s_j;\ O_{j,X}$ & \raggedright Decreasing support weights and local coefficients of $O_j$, with $K^j\rho^{1/2}=O_j\rho^{1/2}$. & \raggedright \eqref{eq:gm:explicit-support-weights}, \eqref{eq:gm:explicit-modular-recursion} \tabularnewline
\hline
\raggedright $\Omega,L,H^{(L)},\rho^{(L)}$ & \raggedright Shielding neighborhood around $A$, range cutoff, Hamiltonian with long terms touching the neighborhood removed, and its Gibbs state. & \raggedright Proof of Thm.~\ref{thm:gm:quasilocal}; \eqref{eq:gm:powerlaw-collar} \tabularnewline
\raggedright $w,J_*$ & \raggedright Positive spatial weight and a uniform bound on the weighted interaction sum inside the neighborhood; the integrated modular leakage decays as $e^{-c_\beta wr}$. & \raggedright Lem.~\ref{lem:gm:powerlaw-collar-leakage} \tabularnewline
\raggedright $w_L$ & \raggedright Spatial weight inside the neighborhood: $(\mu/2)L^{\theta-1}$ for stretched-exponential interactions, and $\log(1+L)/L$ for power-law interactions with $\zeta>2D$. & \raggedright \eqref{eq:gm:quasilocal-collar-weight}, \eqref{eq:gm:powerlaw-collar-weight} \tabularnewline
\hline
\multicolumn{3}{l}{\textbf{Power-law bounds}} \tabularnewline
\hline
\raggedright $\zeta,J_\zeta$ & \raggedright Power-law decay exponent and weighted aggregate pair norm. & \raggedright \eqref{eq:gm:powerlaw-assumption}; Lem.~\ref{lem:gm:powerlaw-support-tail} \tabularnewline
\raggedright $b_\zeta,\nu_\zeta$ & \raggedright $b_\zeta=1+2\zeta/D$ is the small-resolvent tail exponent; $\nu_\zeta$ is the algebraic CMI exponent, with no logarithmic correction. & \raggedright \eqref{eq:gm:powerlaw-exponents} \tabularnewline
\raggedright $J_0,J_1$ & \raggedright Per-site interaction sums without and with one factor of the support size $|X|$. & \raggedright \eqref{eq:gm:powerlaw-J01} \tabularnewline
\raggedright $T_O(N),\mathcal D$ & \raggedright Coefficient tail $T_O(N)=\sum_{|X|>N}\|O_X\|_\infty$ for a chosen decomposition; $\mathcal D$ satisfies $K(O\rho^{1/2})=\mathcal D(O)\rho^{1/2}$. & \raggedright \eqref{eq:gm:powerlaw-D-map} \tabularnewline

\raggedright $d_A;\ C_\beta$ & \raggedright $d_A=\dim\mathcal H_A$. Constants in the power-law theorem depend only on $\beta,D,J_\zeta,\zeta$ and the fixed on-site dimension bound; all region-size dependence is explicit. & \raggedright Thm.~\ref{thm:gm:powerlaw} \tabularnewline
\hline
\end{tabular}
\end{center}
\endgroup
\clearpage

\makeatletter
\immediate\write\@auxout{\string\citation{GibbsBibliographyControl}}
\makeatother
\bibliographystyle{apsrev4-2}
\bibliography{bibilography}
\end{document}